\documentclass{article}
\usepackage[utf8]{inputenc}
\usepackage[letterpaper,margin=.97in]{geometry}
\usepackage{graphicx}
\graphicspath{{img/}}
\usepackage{csquotes}
\MakeOuterQuote{"}

\usepackage{setspace}
\usepackage{url} 
\usepackage{hyperref}
\usepackage{amsthm,amsmath,amssymb}
\usepackage{color}
\usepackage[ruled,vlined]{algorithm2e}
\usepackage{tabularx}
\usepackage{pdflscape,latexsym,bm,array,caption,textcomp}
\usepackage{xspace}
\newcommand{\zell}[2][c]{ \begin{tabular}[#1]{@{}c@{}}#2\end{tabular}}
\newcolumntype{e}{>{\centering\arraybackslash}p{10.6em}}
\newcolumntype{w}{>{\centering\arraybackslash}p{15.5em}}
\newcolumntype{t}{>{\centering\arraybackslash}p{14.5em}}

\newtheorem{theorem}{Theorem}[section]

\newtheorem{lemma}[theorem]{Lemma}
\newtheorem{claim}[theorem]{Claim}

\theoremstyle{definition}
\newtheorem{definition}{Definition}[section]
\newtheorem{observation}[theorem]{Observation}

\newcommand{\A}{{\cal A}}
\newcommand{\C}{{\cal C}}
\newcommand{\D}{{\cal D}}

\newcommand{\F}{{\cal F}}
\newcommand{\G}{{\cal G}}
\renewcommand{\H}{{\cal H}}
\renewcommand{\L}{{\cal L}}
\renewcommand{\O}{{\cal O}}
\renewcommand{\P}{{\cal P}}
\newcommand{\Q}{{\cal Q}}

\newcommand{\cv}{{\cal CV}}
\newcommand{\dist}{\mathrm{dist}}

\newcommand{\slp}{\textsc{Short-Length-Preserver}}

\newcommand{\ftrs}{\textsc{FTRS}}
\newcommand{\kftrs}{\textsc{$k$-FTRS}}
\newcommand{\ftro}{\textsc{FTRO}}
\newcommand{\kftro}{\textsc{$k$-FTRO}}

\newcommand{\new}{\textcolor{red}{(New)}}
\newcommand{\freq}{\textsc{freq}}
\newcommand{\first}{\textsc{first}}
\newcommand{\last}{\textsc{last}}

\title{Pairwise Reachability Oracles and Preservers under Failures}  

\author{}
\author{Diptarka Chakraborty\footnote{National University of Singapore, Singapore.  Supported in part by NUS ODPRT Grant, WBS No. R-252-000-A94-133. Email: diptarka@comp.nus.edu.sg} \and Kushagra Chatterjee \footnote{National University of Singapore, Singapore.  Supported in part by NUS ODPRT Grant, WBS No. R-252-000-A94-133. Email: e0823067@u.nus.edu} \and Keerti Choudhary\footnote{Indian Institute of Technology Delhi, India. Email: keerti@cse.iitd.ac.in}}

\date{}

\newcommand{\Todo}[1]{\textcolor{magenta}{[\textbf{TODO}: #1]}}

\begin{document}
\pagenumbering{gobble}
\maketitle

\begin{abstract}
In this paper, we consider reachability oracles and reachability preservers for directed graphs/networks prone to edge/node failures. Let $G = (V, E)$ be a directed graph on $n$-nodes, and $\P\subseteq V\times V$ be a set of vertex pairs in $G$. We present the first non-trivial constructions of single and dual fault-tolerant 
pairwise reachability oracle with constant query time. Furthermore, we provide extremal bounds for
sparse fault-tolerant reachability preservers, resilient to two or more failures.
Prior to this work, such oracles and reachability preservers were widely studied for the special scenario of single-source and all-pairs settings. However, for the scenario of arbitrary pairs, 
no prior (non-trivial) results were known for dual (or more) failures, except those implied from the single-source setting. One of the main questions is whether it is possible to beat the $O(n |\P|)$ size bound (derived from the single-source setting) for reachability oracle and preserver for dual failures (or $O(2^k n|\P|)$ bound for $k$ failures). We answer this question affirmatively. Below we summarize our contributions.
%not much research has been done. 
 
\begin{itemize}
\item For an $n$-vertex directed graph $G = (V, E)$ and $\P \subseteq V\times V$,  we
present a construction of $O(n \sqrt{|\P|})$ sized dual fault-tolerant 
 pairwise reachability oracle with constant query time.
We further provide a matching (up to the word size) lower bound of $\Omega(n \sqrt{|\P|})$ 
on the size (in bits) of the oracle for the dual fault setting,
thereby proving that our oracle is (near-)optimal.

\item Next, we provide a construction of $O(n + \min\{|\P|\sqrt n,~n\sqrt{|\P|}\})$ sized oracle with $O(1)$ query time, resilient to single node/edge failure. In particular, for $|\P|$ bounded by $O(\sqrt{n})$ this yields an oracle of just $O(n)$ size. We complement the upper bound with a lower bound of $\Omega(n^{2/3}|\P|^{1/2})$ (in bits), refuting the possibility of a linear-sized oracle for $\P$ of size $\omega(n^{2/3})$.
 
 \item We also present a construction of $O(n^{4/3} |\P|^{1/3})$ sized pairwise reachability preservers resilient to dual edge/vertex failures. Previously, such preservers were known to exist only under single failure and had $O(n+\min\{|\P|\sqrt{n},~n\sqrt {|\P|}\})$ size [Chakraborty and Choudhary, ICALP'20].
We also show a lower bound of $\Omega(n \sqrt{|\P|})$ edges on the size of dual fault-tolerant reachability preservers, thereby providing a sharp gap between single and dual fault-tolerant reachability preservers for $|\P|=o(n)$. 

\item Finally, we provide a generic pairwise reachability preserver construction that provides a $o(2^k n |\P|)$ sized subgraph resilient to $k$ failures, for any $k \ge 1$. Before this work, we only knew of an $O(2^k n |\P|)$ bound implied from the single-source setting [Baswana, Choudhary, and Roditty, STOC'16].
 \end{itemize}

\end{abstract}

\newpage
\pagenumbering{arabic}

\section{Introduction}

Networks in most real-life applications are prone to failures. These failures, though unpredictable, 
are transient due to some simultaneous repair process that is undertaken in the
application. This motivates the research on designing fault-tolerant structures for various
graph problems. In the past few years, a lot of work has been done in designing fault-tolerant
structures for various graph problems
like connectivity~\cite{PP14, Parter15, BCR16, GK17, ChakrabortyC20}, finding shortest paths~\cite{DTCR08}, 
graph-structures preserving approximate distances~\cite{Luk99, CZ04, CLPR09, DK11, BK13, BGGLP15, BGPW17, BCHR18} etc. 
Reachability is one of the fundamental graph properties which is as ubiquitous as graphs themselves. 
In this paper, we study pairwise reachability structures under edge/node failures. In particular, 
given any set $\P$ of node-pairs, we provide 
design of graph sparsification structures, and sensitivity oracles for the reachability problem. We present our results in terms of edge failures. However, all our upper bound results also hold for node failures.\footnote{In the input graph, each vertex $v$ can be replaced by an edge $(v_{in},v_{out})$, where all the incoming and outgoing edges of $v$ are directed into $v_{in}$ and directed out of $v_{out}$ respectively.
Thus the failure of vertex $v$ is equivalent to failure of edge $(v_{in},v_{out})$.
}

\subsection{Sensitivity Oracle}
\label{sec:intro-ftro}
In the  Sensitivity oracle, the goal is to design a data structure for a network prone to edge/vertex failures to efficiently answer queries pertaining to the graph structure (e.g., connectivity, reachability, distance, etc.). 
 We first formally define the notion of  Fault-Tolerant Reachability Oracle (FTRO).

\begin{definition}[$\ftro$]
Let $\P\in V\times V$ be any set of pairs of vertices. For a graph $G$, a data structure $DS(G)$ is said 
to be a \emph{$k$-Fault-Tolerant Reachability Oracle} of $G$ for $\P$, denoted as $\kftro(G,\P)$, if 
given a query with any pair $(s,t)\in \P$ and any subset $F\subseteq E$ of at most $k$ edges, $DS(G)$ 
efficiently decides whether or not $t$ is reachable from $s$ in $G \setminus F$. 
\label{definition:FTRO}
\end{definition}

To date, no non-trivial bounds were known for FT-pairwise reachability oracle. 
The only known results are for single-source setting (i.e., $\P=\{s\}\times V$) and all-pairs setting $\P=V\times V$.

For single-source setting, i.e., when $\P=\{s\}\times V$ for some source vertex $s \in V$, under (single and) dual failure, we have an $O(n)$ size oracle with $O(1)$ query time due to~\cite{LT79,Choudhary16}.
As an immediate corollary, for arbitrary $\P$ pairs, we get an $O(n|\P|)$-sized single/dual failure
pairwise $\ftro$ with constant query time. The bound is extremely bad for a large-sized set $\P$. By storing a subgraph that preserves pairwise reachability (to be discussed in detail in Section~\ref{sec:intro-ftrs}) under single failure due to~\cite{ChakrabortyC20}, we get an 1-$\ftro$ of size $O(n+\min\{\sqrt{n} |\P|, n \sqrt{|\P|}\})$ but with $O(n)$ query time. The $O(n)$ query time is due to the fresh reachability computation over the stored subgraph on each query, which is entirely undesirable in terms of the efficiency of a data structure. For the special setting of all-pairs, i.e., $\P=V\times V$, Brand and Saranurak~\cite{BrandS19} provided a $O(n^2)$ sized $\kftro$ that has $O(k^\omega)$ query time, where $\omega$ is the constant of matrix-multiplication. %On the other hand, we get an $O(|\P|)$ sized $O(1)$ time reachability oracle when there is no failure. 

One of the main questions is the following: Does there exist a pairwise reachability oracle of size $o(n |\P|)$ and query time $o(n)$ even for a single failure? In this paper, we answer this question affirmatively. We provide an efficient construction of a $O(n \sqrt{|\P|})$ sized $\ftro$ with constant query time that is resilient to dual failure (not just single failure).
%\keerti{Mention bound here of 1-FTRO.}

\begin{theorem}[Upper Bound on 2-$\ftro$]
\label{thm:dual-oracle}
A directed graph $G=(V,E)$ with $n$ vertices can be processed in randomized polynomial time for a given set $\P \subseteq V \times V$ 
of vertex-pairs, to build a data structure of size $O(n \sqrt{|\P|})$, such that for any pair $(s,t)\in \P$ and 
any set $F$ of (at most) two edge failure, it decides whether there is an $s$ to $t$ path in $G\setminus F$ in time $O(1)$.
\end{theorem}

We further show that the above size bound cannot be improved further by providing a matching (up to the word size) lower bound for two failures. To date, no non-trivial (better than linear) size lower bound is known for any pairwise $\ftro$.

\begin{theorem}[Lower Bound on 2-$\ftro$]
\label{thm:lb-oracle-pair}
For any positive integers $n,r~(r\leq n^2)$, there exists an $n$-vertex directed graph with a vertex-pair set $\P$ of size $r$, such that any 2-$\ftro(G,\P)$ must be of size $\Omega(n\sqrt{|\P|})$ (in bits).
\end{theorem}
In case of source-wise 2-{\ftro} for a source set $S$ (i.e, when $\P=S \times V$), our lower bound construction provides a lower bound of $\Omega(n|S|)$ (in bits). It is again a matching (up to the word size) lower bound because we know of an $O(n|S|)$-sized 2-{\ftro} for any source set $S$ due to~\cite{Choudhary16}. It is also worth noting that our lower bound holds irrespective of the query time and also for directed acyclic graphs.

The above lower bound does not hold for a single failure. So it is natural to ask whether we can design a smaller data structure, more specifically, $O(n)$-sized oracle that is resilient to a single failure. We provide a construction of $O(n + \min\{|\P|\sqrt{n},n\sqrt{|\P|}\})$
sized 1-$\ftro$ with constant $O(1)$ query time. In particular, we show that as long as the number of pairs is bounded by $O(\sqrt{n})$, we can achieve an oracle with $O(n)$ size and $O(1)$ query time. This result provides us a sharp separation in optimum size of a $\ftro$ between single and dual failure. To the best of our knowledge, this is the first separation result between single and dual failure reachability oracle.

\begin{theorem}[Upper Bound on 1-$\ftro$]
\label{thm:single-oracle}
A directed graph $G=(V,E)$ with $n$ vertices can be processed in polynomial time for a given set 
$\P \subseteq V \times V$ of vertex-pairs, to build a data structure of size $O(n + \min\{|\P|\sqrt{n},n\sqrt{|\P|}\})$, such that for any pair $(s,t)\in \P$ and a failure edge $f$, it decides whether there is an $s$ to $t$ path in $G\setminus \{f\}$ in time $O(1)$.
\end{theorem}
%We want to emphasize that in the above theorem, the $O(n\sqrt{|\P|})$ bound is due to Theorem~\ref{thm:dual-oracle}, and also that is the only place where the randomization is used. The construction of $O(n+|\P|\sqrt{n})$ sized data structure is purely deterministic and requires $O(\log n)$ query time. Further, 
Note, the size bound of the above theorem matches the current best known bound for the pairwise reachability preserving subgraph for single failure~\cite{ChakrabortyC20}.
%In a special case, when the input graph is a directed acyclic graph (DAG), we manage to improve our oracle construction to reduce the query time to $O(1)$ (see Section~\ref{section:single-ftro}).

The above upper bound gives $O(n)$ sized oracle only when the number of pairs is $O(\sqrt{n})$. Is it always possible to get a linear-sized pairwise 1-$\ftro$? More specifically, does any $n$-node graph $G$ and a set $\P$ of node-pairs always possess a 1-$\ftro(G,\P)$ of size $O(n + |\P|)$?\footnote{The presence of $|\P|$ term in the bound is justifiable by the fact that for non-failure case (i.e., the standard static setting), we can get a trivial $O(|\P|)$ sized oracle.} In this paper, we refute this possibility by showing the following. 

\begin{theorem}[Lower Bound on 1-$\ftro$]
\label{thm:1-ftro-lb}
For any positive integers $n,d \ge 2$, any $p=p(n)$, there exists an $n$-vertex directed graph and a node-pair set $\P$ of size $p$, such that any 1-$\ftro(G,\P)$ must be of size $\Omega\big(n^{2/(d+1)}p^{(d-1)/d}\big)$ (in bits).
\end{theorem}
By setting $d=2$ in the above theorem, for $p=O(n)$, we get a lower bound of $\Omega(n^{2/3}p^{1/2})$. This shows us that for $p=\omega(n^{2/3})$, there is an $n$-node graph $G$ and a pair set $\P$ of size $p$, for which linear size 1-$\ftro$ is not possible. Again, our lower bound holds irrespective of the query time and also for DAGs. We show the lower bound by establishing a connection between the optimal sized pairwise 1-$\ftro$ and pairwise reachability preserving subgraph without any failure. In general, we show that the optimal size of any pairwise $\kftro$ must be at least that of the reachability preserving subgraph with $k-1$ failures (see Section~\ref{sec:lb-ftro-ftrs}). Instead of just deciding the reachability between a pair of vertex, suppose the data structure is also asked to report a path between them (if exists). Then by a standard information-theoretic argument, the optimal size of any such data structure resilient to $k$ failures must be of size at least that of reachability preserving subgraph with $k$ failures. Unfortunately, such a direct argument does not work for a (Boolean) data structure that only decides the reachability. Ours is the first such connection. Readers may note that there is a gap between our upper and lower bound for pairwise 1-$\ftro$. We leave this as an interesting open question.

\subsection{Reachability Preservers}
\label{sec:intro-ftrs}
In the context of graph sparsification, {\em reachability preserver} (or {\em reachability subgraph})
for a directed graph $G$ and a set $\P$ of vertex-pairs is a sparse subgraph $H$ with as few edges as possible so that for any pair $(s,t) \in \P$ there is a path from $s$ to $t$ in $H$ if and only if there is such a path in $G$. In the standard static setting (with no failure), this object has been studied widely~\cite{CE06, Bodwin17, AB18}. We study these objects in the presence of edge/node failures.

Let us formally define fault-tolerant reachability subgraph ({\ftrs}) for a set of node-pairs.
\begin{definition}[$\ftrs$]
Let $\P\in V\times V$ be any set of pairs of vertices. A subgraph $H$ of $G$ is said to be a \emph{$k$-Fault-Tolerant 
Reachability Subgraph} of $G$ for $\P$, denoted as $\kftrs(G,\P)$, if for any pair $(s,t)\in \P$ and for any subset 
$F\subseteq E$ of at most $k$ edges, $t$ is reachable from $s$ in $G \setminus F$ if and only if $t$ is reachable from 
$s$ in $H \setminus F$.
\label{definition:FTRS}
\end{definition}

%For the non-faulty setting, %The general scenario of $\P$ pairs was very recently studied by 
%Abboud and Bodwin~\cite{AB18},
%%wherein the authors provide 
%gave a construction a pair-wise reachability preserver of size $O(n+(n|\P|)^{2/3} )$.

For the particular case of single-source, i.e., $\P=\{s\}\times V$, 
Baswana, Choudhary, and Roditty~\cite{BCR16} provided a polynomial-time algorithm that, given any $n$-node directed graph, constructs an $O(2^k n)$-sized $\kftrs$. As a corollary, to preserve reachability between arbitrary $\P$ pairs, we get an $O(2^kn|\P|)$-sized $k$-fault-tolerant reachability preserver. For the general setting of arbitrary pairs, the only previously known non-trivial result was for single failure~\cite{ChakrabortyC20}, wherein the authors gave an upper bound of $O(n+\min(|\P|\sqrt{n},~n\sqrt {|\P|}))$ edges. It was left open whether for dual or more failures whether keeping fewer than $O(n|\P|)$ edges sufficient to preserve the pairwise reachability. In particular, does any $n$-node graph and a set $\P$ of node-pairs always admit a $\kftrs$ of size $o(2^k n |\P|)$?

In this work, we answer the above question affirmatively. For dual failures, we provide an upper bound of $O(n^{4/3} |\P|^{1/3})$ edges on the structure of 2-$\ftrs(G,\P)$.

\begin{theorem}[Upper Bound on 2-$\ftrs$]
\label{thm:dual-ftrs}
For any directed graph $G=(V,E)$ with $n$ vertices and a set $\P \subseteq V \times V$ of vertex-pairs, there exists a $2$-$\ftrs(G,\P)$ having at most $O(n^{4/3} |\P|^{1/3})$ edges. Furthermore, we can find such a subgraph in polynomial time.
\end{theorem}
Clearly, for $\P$ of size $\omega(\sqrt{n})$, the above result breaks below the $O(n |\P|)$ bound. We complement our upper bound result the following lower bound.
\begin{theorem}[Lower Bound on 2-$\ftrs$]
\label{thm:lb-ftrs-pair}
For every $n,r~(r\leq n^2)$, there exists an $n$-vertex directed graph $G$ and a vertex-pair set $\P$ of size $r$ such that any $2$-$\ftrs(G,\P)$ requires $\Omega(n\sqrt{|\P|})$ edges.
\end{theorem}
Again, we show a lower bound for source-wise 2-$\ftrs$ for a source set $S$ (i.e, when $\P=S \times V$), of $\Omega(n|S|)$. This matches the $O(n|S|)$ upper bound~\cite{BCR16} of 2-$\ftrs$ for any source set $S$. So for the source-wise preserver, we completely resolve the question regarding the size of an optimal preserver resilient to two or more failures. For general $k>2$ failures, we provide a lower bound of $\Omega(2^{k/2}n \sqrt{|\P|})$ on the size of pairwise {\kftrs} (Theorem~\ref{thm:lb-kftrs-source}).

Previously, seemingly a much weaker lower bound was known~\cite{ChakrabortyC20}, where the authors could only show a lower bound of $\Omega(n|\P|^{1/8})$, for $\P$ of size $n^{\epsilon}$ with $\epsilon \le 2/3$. Our result provides a sharp separation between single and dual fault-tolerant reachability preservers for 
any $\P$ satisfying $\omega(1)\leq|\P|\leq o(n)$. 

We also consider the question of beating $O(2^k n|\P|)$ bound for general $\kftrs$. We show that for a certain regime of the size of $\P$, it is indeed possible to attain $o(2^k n |\P|)$ bound.

\begin{theorem}[Upper Bound on $\kftrs$]
\label{thm:ub-k-ftrs}
For any $k \ge 1$, a directed graph $G=(V,E)$ with $n$ vertices and a set $\P \subseteq V \times V$ of vertex-pairs satisfying $|\P|=\omega(k n^{1-\frac{1}{k}} \log n)$, there exists a $\kftrs(G,\P)$ having only $o(2^k n|\P|)$ edges.
\end{theorem}
We summarize our results on single and dual failures in Table~\ref{table:1}. Readers may note that there is a gap between the size of 2-$\ftro$ and 2-$\ftrs$ in our results. We pose closing this embarrassing gap as an interesting open question.

\begin{table}[!ht]
\begin{center}
\def\arraystretch{1.4}
\begin{tabular}{|e|w| t|}
\hline \rule{0pt}{12pt}
	 \textbf{Problem} & \textbf{Single Failure} & \textbf{Dual Failure} 
\\ \hline	 \rule{0pt}{32pt}
		\zell{Reachability\\ Oracle}%$\ftro$} 
	& \zell{ $O(n + \min(|\P|\sqrt{n},$ $n\sqrt{|\P|}))$ \\ $\Omega(n^{2/3}|\P|^{1/2})$ (in bits) \\ \new}
	& \zell{ $O(n \sqrt{|\P|})$\\$\Omega(n\sqrt{|\P|})$ (in bits)\\ \new}	 
\\ \hline
	\zell{Reachability\\Preserver} %$\ftrs$} 
	& \zell{$O(n+\min(|\P|\sqrt{n},~n\sqrt {|\P|}))$ \\ \cite{ChakrabortyC20}}
	& \zell{$O(n^{4/3} |\P|^{1/3})$ \\$\Omega(n\sqrt{|\P|})$ \\ \new}
\\ \hline		
\end{tabular}
\caption{A comparison of size of $\ftro$ and $\ftrs$ for single and dual failures.}
\label{table:1}
\end{center}
\end{table}

%\begin{table}[!ht]
%\begin{center}
%\def\arraystretch{.95}
%\begin{tabular}{|e|w| t|}
%\hline
%	\bf & \zell {Single Failure} & \bf \zell{Dual Failure}
%\\ \hline	
%		\zell{$\ftro$} 
%	& \zell{ $O(n \sqrt{|\P|} \log n)$\\$O(n + \min\{n^{1/3}|\P|^{4/3},$ $n\sqrt{|\P|}\log n\})$, for DAGs \\(New)}
%	& \zell{ $O(n \sqrt{|\P|} \log n)$\\(New)}	 
%\\ \hline
%	\zell{\\$\ftrs$\\{ }} 
%	& \zell{$O(n+\min(|\P|\sqrt{n},~n\sqrt {|\P|}))$ edges\\ \cite{ChakrabortyC20}}
%	& \zell{$O(n^{4/3} |\P|^{1/3} \log n)$ edges\\(New)}
%\\ \hline		
%\end{tabular}
%\caption{A comparison of Upper bound results on size of $\ftro$ and $\ftrs$ for Single and Dual failures.}
%\label{table:1}
%\end{center}
%\end{table} 

\subsection{Related Work}
A simple version of reachability preserver is when there is a single source vertex $s$, and we would like to preserve reachability from $s$ to all other vertices. Baswana \emph{et al.}~\cite{BCR16} provided an efficient construction of a $k$-fault-tolerant single-source reachability preserver of size $O(2^kn)$. Further, they showed that this upper bound on the size of a preserver is tight up to some constant factor. As an immediate corollary, we get a $\kftrs$ of size $O(2^kn|\P|)$ (by applying the algorithm of~\cite{BCR16} to find subgraph for each source vertex in pairs of $\P$ and then taking the union of all these subgraphs). We do not know whether this bound is tight for general $k$. However, for the standard static setting (with no faulty edges) much better bound is known. We know that even to preserve all the pairwise distances, not just reachability, there is a subgraph of size $O\big(n+\min(n^{2/3}|\P|, n\sqrt{|\P|})~\big)$~\cite{CE06, Bodwin17}. 
Later Abboud and Bodwin~\cite{AB18} showed that for any directed graph $G=(V,E)$ given a set $S$ of source vertices and a pair-set $\P\subseteq S \times V$ we can construct a pairwise reachability preserver of size 
$O\big(n+\min(\sqrt{n |\P| |S|}, (n|\P|)^{2/3})~\big)$. 
It is further shown that for any integer $d \ge 2$ there is an infinite family of $n$-node graphs and vertex-pair sets $\P$ for which any pairwise reachability preserver must be of size $\Omega\big(n^{2/(d+1)}|\P|^{(d-1)/d}\big)$. Note, for undirected graphs, storing spanning forests is sufficient to preserve pairwise reachability information, and thus we can always get a linear size reachability preserver for undirected graphs. We would like to emphasize that all our results in this paper hold for directed graphs. 

By~\cite{BCR16} we immediately get an oracle of size $O(2^k n)$ for $k$ edge (or vertex) failures with query time $O(2^k n)$. 
For just dual failures, we have an $O(n)$ size oracle with $O(1)$ query time due to~\cite{Choudhary16}. %Recently, Brand and Saranurak~\cite{BrandS19} obtained a $k$-fault-tolerant $\O(n^2)$ sized reachability oracle that has $O(k^\omega)$ query time, where $\omega$ is the constant of matrix-multiplication.

For undirected graphs, the optimal bound of $O(kn)$ edges for $k$-fault-tolerant connectivity preserver directly follows from $k$-edge (vertex) connectivity certificate constructions provided by Nagamochi and Ibaraki~\cite{NagamochiI:92}.
For connectivity oracle, P\v{a}tra\c{s}cu and Thorup~\cite{PatrascuT:07} presented a data structure of $O(m)$ size 
that can handle any $k$ edge failures in $O(k\log^2n\log\log n)$ time to subsequently answer 
connectivity queries between any two vertices in $O(\log\log n)$ time.
For small values of $k$,
Duan and Pettie~\cite{DuanP:10} improved the update time of \cite{PatrascuT:07} to $O(k^2\log\log n)$ by
presenting a data structure of $\widetilde O(m)$ size.
For handling vertex failures, Duan and Pettie~\cite{DuanP:17} provided a data structure
of $O(mk\log n)$ size with $O(k^3 \log^3n)$ update time and $O(k)$ query time. 

Other closely related problems that have been studied in the fault-tolerant model include computing distance preservers~\cite{DTCR08, PP13, Parter15}, depth-first-search tree~\cite{BCCK19}, spanners~\cite{CLPR09, DK11}, approximate distance preservers~\cite{BK13, PP14, BGLP16}, approximate distance oracles~\cite{DP09, CLPR10}, compact routing schemes~\cite{CLPR10, Chechik13}.

%For the simpler case of {s}xV and single failure, an FTRO can be obtained using the dominator tree from the work of Lengauer and Tarjan [22] on dominators.

\subsection{Technical Overview}
\paragraph*{Construction of a 2-{\ftrs} for a single pair. }Our starting point is a simple construction of a linear (in the number of vertices) sized 2-{\ftrs} for a single pair. Recall, we already know of such a subgraph by~\cite{BCR16}. However, this new alternate construction will shed more light on the specific structure of a 2-{\ftrs}, which will play a pivotal role in our constructions of pairwise 2-{\ftrs} and 2-{\ftro}. Given a directed graph $G$ and a vertex-pair $(s,t)$, we construct a subgraph $H_{(s,t)}$ as follows: First, consider two "maximally disjoint" $s-t$ paths $P_{(s,t)}^1$ and $P_{(s,t)}^2$ (that meet only at the cut-edges and cut-vertices). We refer to these two paths as \emph{outer strands}. Next, we add several \emph{coupling paths} between these two outer strands, which are edge-disjoint with the outer strands. For each vertex $v$ on the outer strands, we check for the "earliest" vertex on the strand $P_{(s,t)}^1$ (and $P_{(s,t)}^2$), from which there is a path $Q_{(s,t),v}^1$ (and $Q_{(s,t),v}^2$) to $v$ that is edge-disjoint with both the outer strands. We refer to these path $Q_{(s,t),v}^i$ %s 
as coupling paths. Roughly speaking, two outer strands together with the coupling paths constitute the subgraph $H_{(s,t)}$ (see Figure~\ref{fig:single-pair-ftrs}). The actual construction is slightly different. Let us first briefly discuss why the above subgraph is a 2-{\ftrs} for the pair $(s,t)$. Then we will comment on the issue with the above simple construction and how we overcome that.

\begin{figure}[!ht]
\centering
\includegraphics[width=200pt,height=300pt,keepaspectratio]{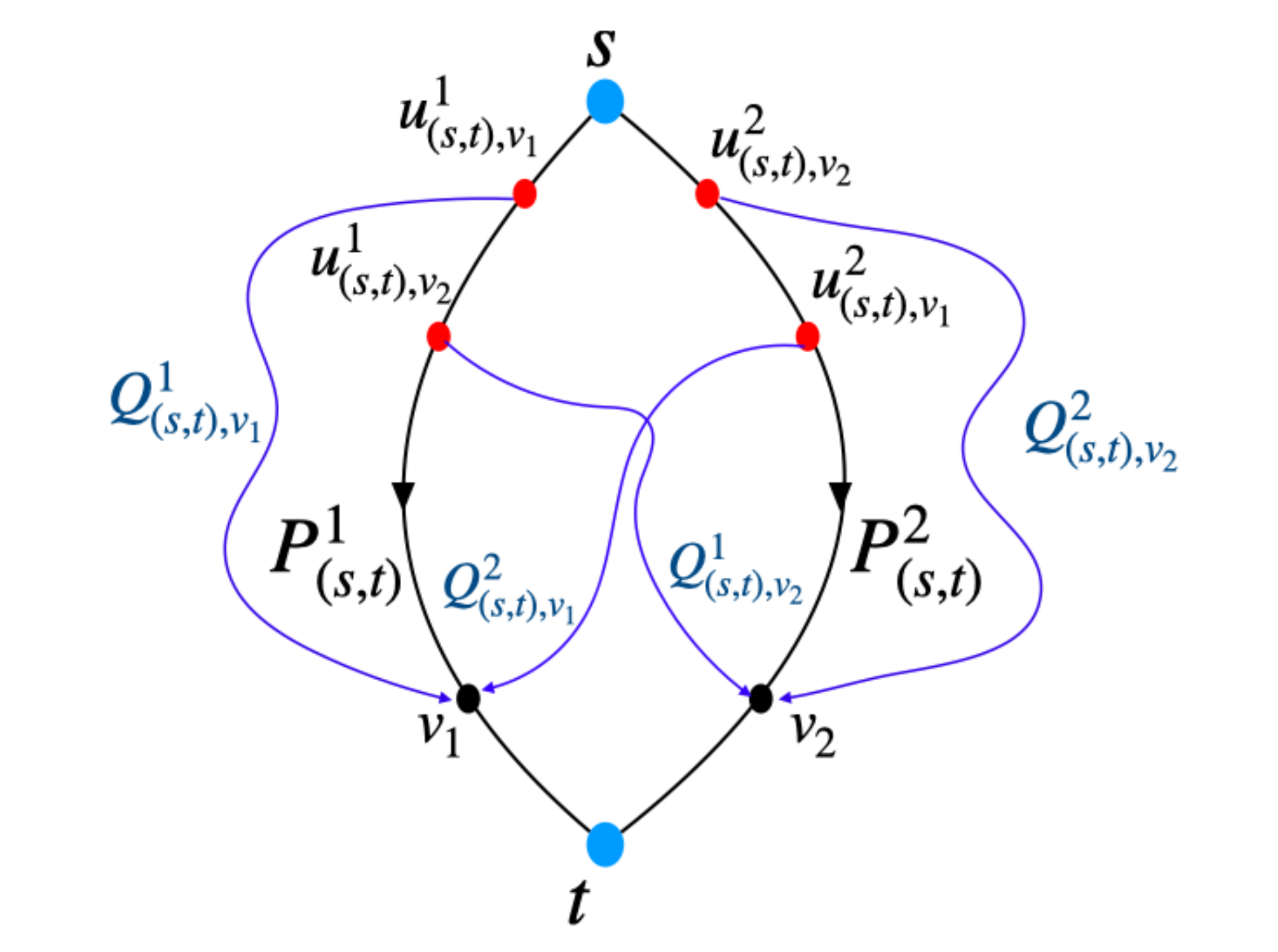}
\caption{$H_{(s,t)}$=2-{\ftrs} for a single pair $(s,t)$. Two black paths are the outer strands and the purple paths are the coupling paths between them.}
\label{fig:single-pair-ftrs}
\end{figure}

Consider any two failure edges $f_1$, $f_2$. W.l.o.g. assume, they do not form an $s-t$ cut-set; otherwise, after the failure there won't be any $s-t$ path. Thus if both $f_1,f_2$ lie on one of the two outer strands (i.e., either on $P_{(s,t)}^1$ or $P_{(s,t)}^2$), then since these two strands are maximally disjoint, one of them will survive after the failures. So, let $f_1$, $f_2$ lie on the strand $P_{(s,t)}^1$, $P_{(s,t)}^2$ respectively. Then consider the subpaths of $P_{(s,t)}^1$, $P_{(s,t)}^2$ above $f_1$, $f_2$, and the subpaths of $P_{(s,t)}^1$, $P_{(s,t)}^2$ below $f_1$, $f_2$. Since by assumption $f_1, f_2$ does not form an $s-t$ cut-set, there must be a coupling path (edge-disjoint with $P_{(s,t)}^1$, $P_{(s,t)}^2$) from one of the top subpaths to one of the bottom subpaths in $G$. Since $H_{(s,t)}$ consists of all the coupling paths, we get a surviving path in $H_{(s,t)}\setminus \{f_1,f_2\}$. This shows that $H_{(s,t)}$ is a 2-$\ftrs(G,(s,t))$. Moreover, one may observe from the above argument that, after failure of any two edges, one of the surviving paths in $H_{(s,t)}$ must be of the following form: It first follows one of the outer strand from $s$ to some vertex $u$, then takes a coupling path till some vertex $v$ on one of the outer strands, and finally follows the corresponding outer strand from $v$ to $t$. We refer to such a path as \emph{nice path}. The existence of such nice paths helps us in proving the correctness of our pairwise 2-{\ftro} and 2-{\ftrs} construction in the subsequent sections.

As we mentioned earlier, our actual construction is slightly different. The main issue with the above simple construction is that the constituted subgraph could be of size $\omega(n)$ after adding all the coupling paths. To mitigate this issue, instead of adding all the coupling paths, we only add the "essential" coupling paths. (See Section~\ref{section:single-pair} for details.) It allows us to achieve $O(n)$ size bound without affecting the correctness of 2-{\ftrs}. The guarantee of the existence of nice paths also remains unaffected. Of course, the correctness argument will become slightly more intricate.

Next, we use the above construction of a 2-{\ftrs} of a single pair to study the pairwise dual fault-tolerant graph structures (reachability oracle and preserver). Our input is a directed graph $G=(V,E)$ with $n$ nodes, and a node-pair set $\P \subseteq V \times V$.

\paragraph*{Pairwise 2-{\ftro}: Upper bound. }
One of the main contributions of this work is a construction of a dual fault-tolerant pairwise reachability oracle (2-{\ftro}) of size $O(n \sqrt{|\P|})$. For simplicity, below, we briefly describe a construction that provides a slightly weaker bound, in particular, $O(n \sqrt{|\P|} \log n)$. Later we will comment on how to remove this extra $\log n$ factor.

We start with the 2-{\ftrs} $H_{(s,t)}$, for each pair $(s,t) \in \P$. First, for all $(s,t)\in \P$, we consider the top and bottom $\Theta(n/\sqrt{|\P|})$ portion of the outer strands ($P_{(s,t)}^1$, $P_{(s,t)}^2$). We find a subset of vertices that intersects all these subpaths, i.e., acts as a "hitting set". Using a standard greedy algorithm we get a hitting set of size $O(\sqrt{|\P|}\log n)$. (Note, one may alternatively use random sampling to achieve the same bound for the hitting set with high probability.) Then we compute linear-sized single-source and single-destination 2-{\ftro} having query time $O(1)$, for each of the vertices in the hitting set using~\cite{Choudhary16}. For each $(s,t) \in \P$, in a table $T_{(s,t)}$, we store one vertex from each of the top and bottom $\Theta(n/\sqrt{|\P|})$ length subpaths of $P_{(s,t)}^1$, $P_{(s,t)}^2$, that is also included in the hitting set. That constitutes the first part of our data structure.

Observe, for any two failure edges $f_1$, $f_2$ and a pair $(s,t) \in \P$, we know that there must be a surviving \emph{nice} $s-t$ path in $H_{(s,t)}\setminus \{f_1,f_2\}$ (unless $f_1$, $f_2$ form an $s-t$ cut-set). Now, if that surviving path follows the bottom (or top) $\Theta(n/\sqrt{|\P|})$ length subpath of $P_{(s,t)}^j$ (for some $j \in \{1,2\}$), it must also pass through one of the stored vertices in $T_{(s,t)}$, say $v$ (due to the hitting set property). In this scenario, we can easily check for presence of an $s-v$ and $v-t$ path after the failure of $f_1$, $f_2$ in $G$ by trying with all the (at most four) stored vertices in $T_{(s,t)}$. For that, we only need $O(1)$ time during query.

Now, it only remains to consider the case when the surviving nice path in $H_{(s,t)}\setminus \{f_1,f_2\}$ passes though a coupling path $Q_{(s,t),v}^i$, for some vertex $v$ lying on the bottom $\Theta(n/\sqrt{|\P|})$ length subpath of a outer strand, that starts from some vertex $u$ lying on the top $\Theta(n/\sqrt{|\P|})$ length subpath of a outer strand. Informally, we only need to consider the scenario when both the outer strands are of length $O(n /\sqrt{|\P|})$. For simplicity, from now on we continue the description with this assumption. Intuitively, this enables us to look into a smaller graph (which need not be a subgraph of the original graph). We build an \emph{auxiliary graph} $A_{(s,t)}$ (see Figure~\ref{fig:auxiliary}) from $H_{(s,t)}$. We define the auxiliary graph entirely on the vertex set of the outer strands ($P_{(s,t)}^1$, $P_{(s,t)}^2$). First, we add both the outer strands (i.e. all their edges) to the auxiliary graph. Next, we add \emph{auxiliary edges} between the vertices if and only if there is a path between them in $H_{(s,t)}$ that is edge-disjoint with the outer strands. Then, We construct a 2-{\ftro} (having query time $O(1)$) for $A_{(s,t)}$ using~\cite{Choudhary16}. We do this for all pair $(s,t) \in \P$. Since each auxiliary graph is defined over a set of $O(n/\sqrt{|\P|})$ sized vertex set, we need total $O(n\sqrt{|\P|})$ space. This finishes the description of our data structure. So, the final data structure consists of 2-{\ftro}s of the vertices in the hitting set and 2-{\ftro} computed over $A_{(s,t)}$'s. Hence, the size of the whole data structure is $O(n\sqrt{|P|} \log n)$.

\begin{figure}[!ht]
\centering
\includegraphics[width=300pt,height=600pt,keepaspectratio]{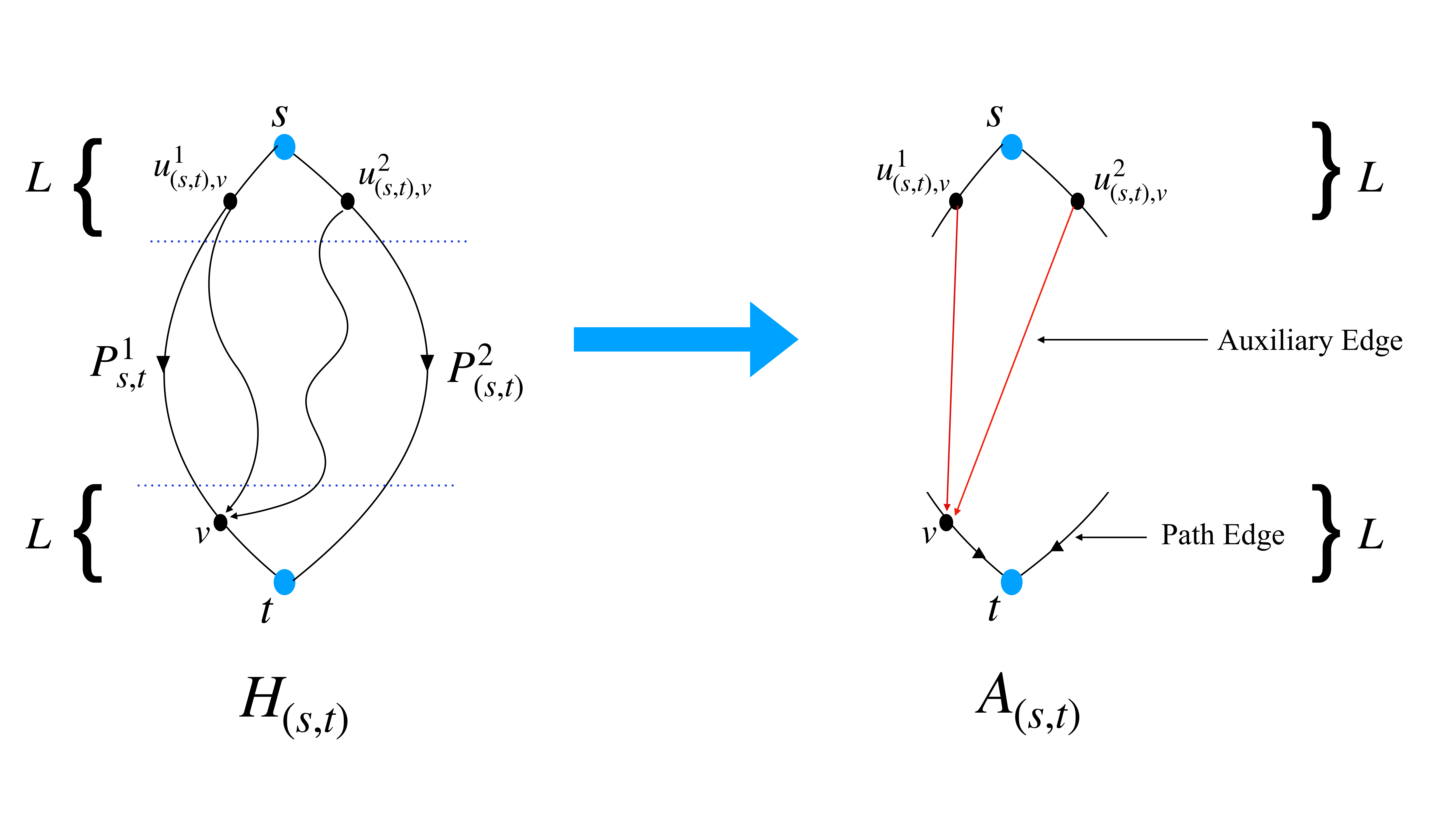}
\caption{Auxiliary graph $A_{(s,t)}$ constructed from $H_{(s,t)}$}
\label{fig:auxiliary}
\end{figure}

It is not hard to see that any $s-t$ path of $H_{(s,t)}$ also leads to a valid $s-t$ path in $A_{(s,t)}$ and vice versa. However, it is not immediate that it will be the case even after the failure of $f_1$, $f_2$. The difficulty arises because many paths in $H_{(s,t)}$ now map to one path in $A_{(s,t)}$. We show that it is indeed the case that there is an $s-t$ path in $H_{(s,t)}\setminus \{f_1,f_2\}$ if and only if there is an $s-t$ path in the auxiliary graph $A_{(s,t)}\setminus \{f_1,f_2\}$ (given $P_{(s,t)}^1$, $P_{(s,t)}^2$ are of length at most $O(n /\sqrt{P})$). The actual description is slightly more involved because we cannot make any assumption on the length of $P_{(s,t)}^1$, $P_{(s,t)}^2$. Essentially, we need only to consider the top and bottom $O(n \sqrt{|\P|})$ portion of the outer strands and define the auxiliary graph over them. Then we prove the above claim without any assumption on the length of the outer strands. The guarantee on the existence of a nice $s-t$ path in $H_{(s,t)}$ after at most two failures comes handy in this case. Recall, in a nice path, there is at most one coupling sub-path. If both the endpoints of this coupling sub-path lie on the top and bottom $O(n \sqrt{|\P|})$ portion of the outer strands, we get an auxiliary edge. As a result, we get an $s-t$ path in the auxiliary graph after the failures.  We refer the readers to Section~\ref{section:dual-ftro} for the details. Note, without the guarantee of a nice path, there could be many coupling sub-paths in a surviving $s-t$ path after failures. As a result, we may not get an auxiliary edge in our auxiliary graph.

During the query for a pair $(s,t)\in \P$, it suffices to either place $O(1)$ many reachability queries on the first part of our data structure or check the presence of an $s-t$ path in the auxiliary graph $A_{(s,t)}$ (see Lemma~\ref{lem:dual-oracle-correctness}). Hence, we get the overall query time to be only $O(1)$.

To remove the $O(\log n)$ factor from the size bound, we use the concept of \emph{sparsifier with slack}~\cite{chan2006spanners, konjevod2007compact, dinitz2007compact, bodwin2020note}. The extra $O(\log n)$ factor was coming due to the construction of the greedy hitting set. We show that if we prematurely terminate the same greedy hitting set algorithm, we get a \emph{fractional hitting set} that hits a constant fraction of the input sets (see Theorem~\ref{thm:frac-hitting-set}). As a consequence, we get a pairwise 2-{\ftro} with slack of size $O(n\sqrt{|\P|})$. Then we use an argument similar to that in~\cite{bodwin2020note} to get the same size bound for the (standard) pairwise 2-{\ftro}.

\paragraph*{Optimality of pairwise dual fault-tolerant oracle. }
In this paper, we show that for pairwise 2-{\ftro}, our $O(n \sqrt{|\P|})$ bound is essentially tight up to the word size (Theorem~\ref{thm:lb-oracle-pair}). Actually, we show that for any source-wise 2-{\ftro} for a designated source set $S$, the trivial $O(n|S|)$ upper bound followed from~\cite{Choudhary16}, is tight. To do that, we first provide a $\Omega(n|S|)$ lower bound for source-wise 2-{\ftrs} (leading to Theorem~\ref{thm:lb-ftrs-pair}) by constructing a hard instance. Then we extend that lower bound to source-wise 2-{\ftro} using communication complexity.

We construct the hard instance for 2-{\ftrs} as follows. Take any integer $N,r$. We consider two $r$-sized sets of vertex-disjoint (directed) paths each of length $N$. Then include a (directed) complete bipartite graph between the left set of $r$ vertices and the right set of $r$ vertices, for each level $k \in [N]$ (See Figure~\ref{fig:lower-bound}). The set $S$ contains the start vertices of the paths in the left set and the terminal vertices of the paths in the right set. So, $|S|=2r$. The number of vertices and the edges in this graph are $n=2Nr$ and $\Theta(Nr^2)$ respectively. It is not difficult to observe that each edge in this graph must be present in any 2-{\ftrs} for this graph with the pair set $\P=S\times S$ (See Figure~\ref{fig:lower-bound}). Consequently, we get a $\Omega(n|S|)=\Omega(n\sqrt{|\P|})$ lower bound.

\begin{figure}[!ht]
\centering
\includegraphics[scale=.45]{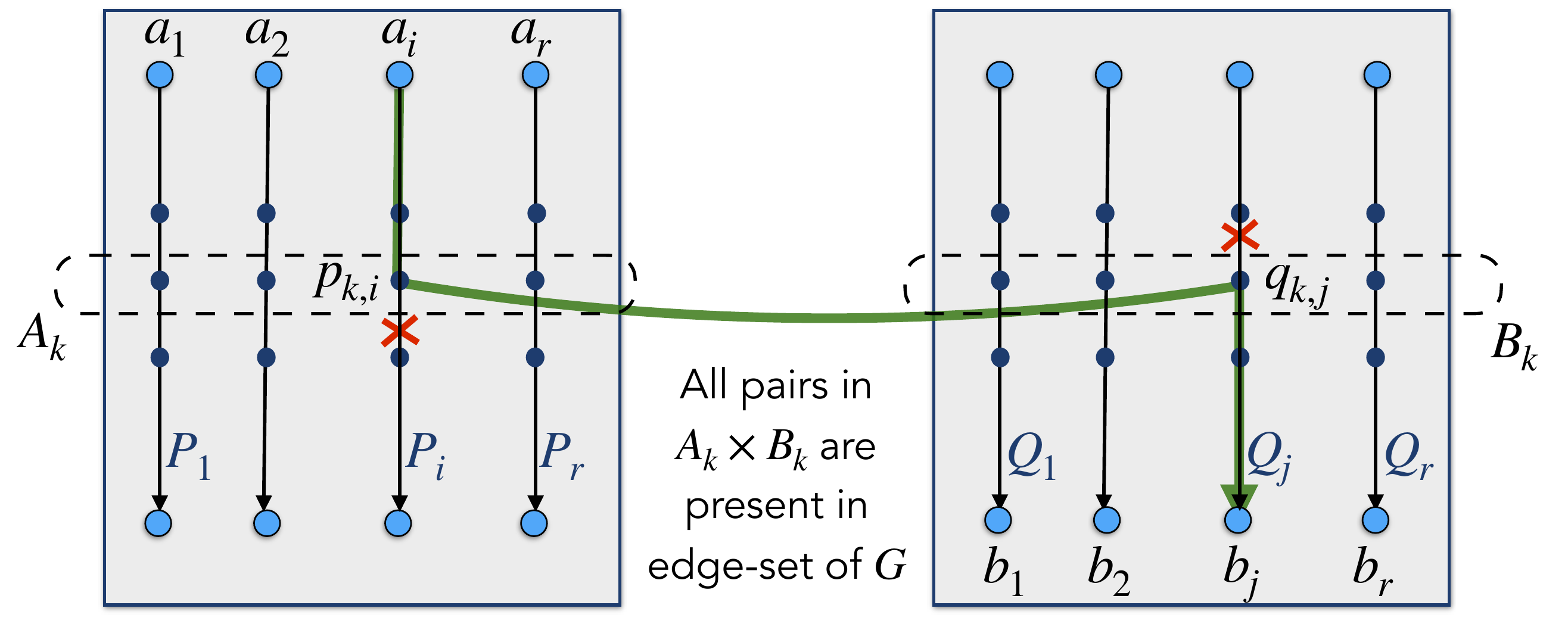}
\caption{Dual fault-tolerant reachability preserver.}
\label{fig:lower-bound}
\end{figure}

We then extend the above construction to 2-{\ftro}. Note, {\ftro} is a Boolean data structure that only decides reachability between a pair of vertices - does not report a path (if exists). If the data structure also reports a path (if exists), then using a standard information-theoretic argument, it is possible to show that the data structure must be of a size equal to that of a {\ftrs}. Such an argument does not work in general for {\ftro}. This is a general hurdle that we need to overcome to extend the lower of {\ftrs} to {\ftro}. Fortunately, the hard instance we constructed for {\ftrs} allows us to provide a reduction from (a close variant of) the \emph{Index} problem. Then we use the randomized one-way communication complexity lower bound for the Index problem~\cite{KNR99} to get a lower bound of $\Omega(n|S|)$ on the size (in bits) of any 2-{\ftro}. It is worth mentioning that our lower bound holds irrespective of the query time of 2-{\ftro}.

\paragraph*{Pairwise 1-{\ftro}: Upper and lower bound. }
The above lower bound only holds for dual failures. So it remains open whether for single failure we can attain better than $O(n\sqrt{|\P|})$ size bound with $o(n)$ (ideally, constant) query time. In this section, we provide a construction of 1-{\ftro} of size $O(n+|\P|\sqrt{n})$ and query time $O(1)$. We show the construction by first designing a reachability oracle of size $O(n+|\P|\sqrt{n})$ resilient to single vertex failure. Then we extend that vertex fault-tolerant data structure to an edge fault-tolerant data structure of the same size.

One of the key ingredients of our construction is a clever data structure that could answer all-pairs reachability query under single vertex failure as long as all three vertices involved in the query belong to a cut-set. Formally, for a pair $(s,t)$, let $C$ be any subset of $s-t$ cut-vertices in the graph $G$. We build an $O(|C|)$ size data structure that, for any three vertices $x,y,z\in C$, decides the reachability between $y,z$ upon failure of %the 
vertex $x$, in constant time. This data structure is inspired by the loop-nesting-forest~\cite{Tarjan76}. We consider the ordering $\sigma$ among the cut-vertices in $C$. Next, we build a predecessor forest and a successor forest with respect to this ordering $\sigma$. In the predecessor forest, the parent of any node $w$ is the immediate predecessor $u$ such that $u, w$ are strongly connected even after the removal of all the predecessors of $u$ from $G$. We symmetrically build the successor forest. Using constantly many Lowest Common Ancestor (LCA) and Level Ancestor (LA) queries on these forests, we can now decide the reachability between $y,z$ upon failure of the vertex $x$ (for any $x,y,z \in C$). By deploying any standard linear space LCA/LA data structure, we attain $O(|C|)$ space-bound and $O(1)$ query time. We refer the readers to Section~\ref{sec:all-pair-cut} for the details.

Next, we use the above data structure to construct a pairwise fault-tolerant reachability oracle for single vertex failure. Observe, to decide the reachability between a pair upon failure, we need to check whether the failure vertex is a cut-vertex for that pair or not. At the high level, we need to keep the information about the cut-vertex-set for each pair. However, we cannot store them explicitly using small space. We first consider the cut-vertex-set for all the pairs in $\P$. Then we identify a \emph{core pair-set} as follows: Take any pair in $\P$ with at least $\sqrt{n}$ cut-vertices, and add it into the core pair-set. Then iteratively add pairs from $\P$ with at least $\sqrt{n}$ new/uncovered (not part of the cut-vertex-set of any previously added pair) cut-vertices. For each pair added in the core pair-set, it \emph{owns} the corresponding cut-vertices that were uncovered till then. Once we get the core pair-set, we are left with pairs each having at most $\sqrt{n}$ uncovered cut-vertices. Further note, the size of the core pair-set is at most $O(\sqrt{n})$. Consider the union of all the cut-vertices of the core pair-set, and build the previously mentioned data structure on it. For the remaining pairs, we store the uncovered cut-vertices associated with them using any static dictionary data structure. Note, each pair might have cut-vertices that are owned by some core pair. To keep track of them, for each pair and each core pair, we store the first and the last cut-vertex shared by them. These all constitute the final data structure. It is not hard to see that the size is $O(n + |\P|\sqrt{n})$. We show that it suffices to either perform a query on the previously mentioned data structure or check whether the failure vertex is a cut-vertex in the dictionary structure during the query. (See Section~\ref{section:FTRO_vertex}.)

Now, we extend the vertex fault-tolerant oracle to edge fault-tolerant oracle. We consider the set $\C$ of all the cut-edges for all the pairs in $\P$. Then we find the subset $\C_0$ that only contains the edges whose endpoints are strongly connected in $G$, but not after the failure of that edge. Observe, all the edges in $\C$ must be part of any strong-connectivity certificate, and all the edges in $\C_0$ must be part of any reachability preserver without any failure. Thus $|\C| = O(n)$ and $|\C_0|=O(n+\min\{|\P|\sqrt{n},(n |\P|)^{2/3}\})$ by~\cite{AB18}. Next, we construct a new graph with $n+|\C_0|$ vertices such that checking the reachability upon an edge failure in the original graph reduces to checking the reachability (perhaps between a different pair) upon vertex failure in the new graph. As a consequence, we get an $O(n + |\P|\sqrt{n})$ sized 1-{\ftro} with constant query time.

We complement our upper bound result with a lower bound of $\Omega(n^{2/{d+1}}|\P|^{(d-1)/d})$ size (in bits) for any $d \ge 2$. We prove our lower bound by establishing a connection between the optimal size of a pairwise {\kftro} and that of a pairwise $(k-1)$-{\ftrs}. In particular, we show that the optimal size of {\ftro} for any $n$-node graph and $p$-sized pair-set is at least that of $(k-1)$-{\ftrs} for any $n/2$-node graph and $p$-sized pair-set (Theorem~\ref{thm:lb-ftrs-ftro}). Such a connection between {\ftro} and {\ftrs} is entirely new. To prove this relation, we use information-theoretic encoding-decoding argument. (See Section~\ref{sec:lb-ftro-ftrs}.) Then the lower bound of 1-$\ftro$ follows from the lower bound of reachability preserver without any failure by~\cite{AB18}.

\paragraph*{Pairwise 2-{\ftrs}: Upper bound. }
We have already shown a lower bound of $\Omega(n\sqrt{|\P|})$ on the size of a pairwise 2-{\ftrs}. However, so far we only known of $O(n |\P|)$ size 2-{\ftrs} for any $n$-node graph $G$ and pair set $\P$. In this work, we provide a deterministic polynomial time construction of a pairwise 2-{\ftrs} of size $O(n^{4/3} |\P|^{1/3})$.  For simplicity, below, we briefly describe a construction that provides a slightly weaker bound, in particular, $O(n \sqrt{|\P|} \log n)$. In the actual construction, we get rid of the $\log n$ factor by constructing a preserver with slack and then using the result of~\cite{bodwin2020note} - an idea similar to that used in pairwise 2-{\ftro} described earlier. Before proceeding further, let us emphasize that the main underlying idea behind our construction could be generalized to {\kftrs}, for any $k \ge 1$, albeit with the help of randomization. We describe that generic construction later.

To get a sparse 2-{\ftrs} for a node-pair set $\P$, we perform two-step sparsification. First, we apply our alternate construction of 2-{\ftrs} for each of the pairs of $\P$. Then take a union of all of these subgraphs to get a $O(n|\P|)$ size intermediate subgraph $H_{inter}$, which is clearly a 2-{\ftrs} for $\P$. Next, we further sparsify this intermediate subgraph. Similar to the technique used in oracle construction, we consider the top and bottom $\Theta(n^{2/3}|P|^{-1/3})$ portions of the outer strands ($P_{(s,t)}^1$, $P_{(s,t)}^2$) of $H_{s,t}$, for each $(s,t)\in \P$. Next, we construct a greedy hitting set containing $O(n^{1/3}|\P|^{1/3} \log n)$ vertices that intersects all these subpaths. We compute linear-sized single-source and single-destination 2-{\ftrs} for each of these vertices in the hitting set using~\cite{BCR16}. Let $H_1$ be the union of all these single-source and single-destination 2-{\ftrs}. So, $H_1$ is of size $O(n^{4/3} |\P|^{1/3} \log n)$.

Consider any two failure edges $f_1,f_2$ and a pair $(s,t)\in \P$. If there is a surviving path in $G\setminus \{f_1,f_2\}$, we know that there is a \emph{nice} $s-t$ path in $H_{(s,t)}\setminus \{f_1,f_2\}$. Using an argument similar to that in the oracle construction, if that nice path follows the top or bottom $\Theta(n^{2/3}|P|^{-1/3})$ portion of a outer strand, then there is also an $s-t$ path in $H_1 \setminus \{f_1,f_2\}$. So now on, it suffices to look into the case when the surviving nice path in $H_{(s,t)}\setminus \{f_1,f_2\}$ passes though a coupling path $Q_{(s,t),v}^i$, for some vertex $v$ lying on the bottom $\Theta(n^{2/3}|P|^{-1/3})$ length subpath of a outer strand, that starts from some vertex $u$ lying on the top $\Theta(n^{2/3}|P|^{-1/3})$ length subpath of a outer strand.

If we could include the top and bottom $\Theta(n^{2/3}|P|^{-1/3})$ portion of the outer strands, and all the "essential" coupling path $Q_{(s,t),v}^i$'s with endpoints lying on the top and bottom $\Theta(n^{2/3}|P|^{-1/3})$ portions of the outer strands, we will be done. We indeed consider a union of the top and bottom $\Theta(n^{2/3}|P|^{-1/3})$ portion of the outer strands  (for all $(s,t) \in \P$), and let us denote that by $H_2$. So, $H_2$ is of size $O(n^{2/3}|P|^{2/3})$. Unfortunately, we do not have a guarantee on the length of the coupling paths. Thus, if we include all the required coupling paths, we cannot argue about the sparsity of the final subgraph. (This portion of the construction differs significantly from that of our oracle construction.) We consider the subgraph obtained by taking a union of all the "essential" coupling paths with endpoints lying on the top and bottom $\Theta(n^{2/3}|P|^{-1/3})$ portions of the outer strands. (Let us denote this union by $B$.) Then we sparsify this subgraph further. For that purpose, we first isolate all the "high frequency" vertices (iteratively) and remove all the coupling paths containing them. Since total number of coupling paths in this subgraph is only $\Theta(n^{2/3}|P|^{2/3})$, we end up with "a few" ($O(n^{1/3}|P|^{1/3})$) high frequency vertices. Now, observe, in the remaining subgraph (denoted as $H_4$), degree of each vertex is "small" (at most $O(n^{1/3}|P|^{1/3})$). Next, for each of the high-frequency vertices, compute linear-sized single-source and single-destination 2-{\ftrs}, and take a union of them to form a subgraph $H_3$. The union of $H_1,H_2,H_3$ and $H_4$ constitute the final subgraph.

It is not difficult to see that a surviving nice $s-t$ path in $H_{(s,t)}\setminus \{f_1,f_2\}$ either passes through one of the high frequency vertices, in which case we get an $s-t$ path in $H_3$; or is included in $H_2 \cup H_4$. Thus  So, the union of all $H_1, H_2, H_3$ and $H_4$ will be a 2-{\ftrs} for $\P$. It is worth mentioning that the correctness proof works only because of a guarantee of the existence of the nice path. Note, a nice path follows at most one coupling path as a subpath. Thus either that nice path follows the top or bottom  $\Theta(n^{2/3}|P|^{-1/3})$ portion of the outer strands, or the coupling sub-path is part of $B$, which we further sparsify to get $H_3$ and $H_4$. Without the guarantee of a nice path, there could be many coupling sub-paths in a surviving $s-t$ path after failures. Endpoints of these coupling sub-paths may not lie on the top or bottom $\Theta(n^{2/3}|P|^{-1/3})$ portion of the outer strands. As a result, we miss them in $B$. As a consequence, $H_3 \cup H_4$ could not capture those coupling sub-paths.

Observe, each of the $H_i$'s is of size at most $O(n^{4/3}|\P|^{1/3} \log n)$ (the $\log n$ factor is only there for $H_1$). So the total size is also $O(n^{4/3}|P|^{1/3} \log n)$.

\paragraph*{Pairwise {\kftrs}: Beating the $O(2^k n|\P|)$ bound. }
Currently, for any $k\ge 1$, we only know there always exists a pairwise {\kftrs} of size $O(2^k n |\P|)$ for any pair set $\P$~\cite{BCR16}. Previously, we beat the above bound for $k=2$. We cannot directly extend that bound for {\kftrs}. The main obstacle is that for pairwise 2-{\ftrs}, we have used a particular structure of a 2-{\ftrs} for a single pair (i.e., the subgraph $H_{(s,t)}$). It is pretty difficult get such kind of structure for {\kftrs} for any $k >2$. However, the main underlying idea behind our pairwise 2-{\ftrs} construction is to handle the "long" and "short" surviving paths separately. Informally, in the case of 2-{\ftrs}, the hitting set helps us to look into only the short paths in $H_{(s,t)}$. We do a similar thing for {\kftrs} using randomization.

Take a parameter $\ell$, whose value we will fix later. We sample a uniformly random subset $W$ of vertices of size $\tilde{O}(kn/\ell)$. Then we build a single-source and single-destination {\kftrs} from those vertices. Let us denote the union of all these {\kftrs}s as $H_1$. The size of $H_1$ is $O(k2^k n^2/\ell)$ by~\cite{BCR16}. Suppose, for any failure edge-set $F$ of size at most $k$, the length of any shortest $s-t$ surviving path in $G\setminus F$ is at least $\ell$. Consider an (arbitrarily chosen) shortest surviving $s-t$ path in $G \setminus F$ (which is of length at least $\ell$). Then it is not hard to argue that $W$ intersects that path with high probability (Lemma~\ref{lemma:hitting_property}). Thus we get an $s-t$ path in $H_1\setminus F$. This part is similar to what we get in 2-{\ftrs} using greedy hitting set.

So now on, we only need to handle the case when a shortest surviving path is of length at least $\ell$. To do this, roughly speaking, we enumerate over all possible failure-set of size at most $k$ and add a shortest surviving path in the subgraph. Trivially, we get a bound of only $O(n^k)$ on the number of possible failure edge-sets of size at most $k$. However, observe, we only need to handle the case when a surviving shortest path is of length at most $\ell$. So before the failure also the length of a shortest path was at most $\ell$. The initial shortest path (before any failure) will get destroyed only if one or more failure edges lie on that shortest path. This observation reduces the number of possible failure-set to $O(\ell^k)$. So, for each pair, we add $O(\ell^k)$ paths, each of length at most $\ell$. For all the pairs, the total number of added edges is at most $O(\ell^{k+1}|\P|)$. We can further improve this bound to $O(\ell^k |\P|)$ by enumerating the failure-set of size up to $k-1$ and then adding two maximally edge-disjoint shortest paths. (See Section~\ref{section:k-ftrs} for the details.) At the end, we get a subgraph of size $O(k2^k n^2/\ell + \ell^k |\P|)$. By optimizing the parameter $\ell$, we get a size bound of $\widetilde O(k~ 2^k ~n^{\frac{2k}{k+1}}~ |\P|^\frac{1}{k+1})$. Note, this size bound beats the $O(2^k n |\P|)$ bound whenever $|\P|=\omega(k n^\frac{k-1}{k} \log n)$.

\subsection{Conclusion}
\label{sec:conclusion}
In this paper, we study compact oracle and sparse preservers for the problem of pairwise reachability under failures. For dual failures, we provide a construction of $O(n\sqrt{|\P|})$ sized reachability oracle
that has constant query time, along with a matching (up to the word size) lower bound. It (almost) settles down the question on the optimal sized dual fault-tolerant pairwise reachability oracle. For single failure, we achieve a bound of $O(n + \min\{|\P|\sqrt{n},n\sqrt{|\P|}\})$ on the size of oracle with constant query time. We complement our upper bound with a $\Omega(n^{2/3}|\P|^{1/2})$ size (in bits) lower bound, refuting the possibility of getting a linear-sized oracle for a single failure. We would like to pose the problem of closing the current gap between the upper and lower bound for a single fault-tolerant oracle as an open problem.

In the case of reachability preserver, we show an upper bound of $O(n^{4/3} |\P|^{1/3})$ edges for any $n$-node graph and a vertex-pair set $\P$, for dual failures setting. This improves the naive bound of $O(n|\P|)$ preserver (obtained from the result known for single-source setting). In addition, we obtain a lower bound of $\Omega(n \sqrt{|\P|})$ on the size of our reachability preserver under dual failures. Prior to our work it was known that for single failure, we can have preserver with $O(n+\min\{|\P|\sqrt{n},~n\sqrt {|\P|}\})$ edges. Thus our lower bound provides a striking difference between the single and dual fault-tolerant setting as it gives a separation of $n^{1/4}$ factor for $|\P|=\Theta(\sqrt n)$. One immediate open question after our work is whether a sparser pairwise reachability preserver exists for dual failures.

Next, we go beyond the dual failures and consider the question of getting sparse pairwise reachability preserver resilient up to $k$ failures for any $k \ge 1$. Can there always exist a pairwise preserver of size $o(2^k n |\P|)$ for any $k$ failures? We answer this question affirmatively by providing a randomized polynomial-time construction of a $o(2^k n|\P|)$ size preserver. We leave the question of finding an optimal $k$ fault-tolerant pairwise preserver as an important open problem.

\section{Preliminaries and Tools}
\label{section:prelims}
For any integer $n$, we use $[n]$ to denote the set $\{1,2,\cdots,n\}$. Given a directed graph $G=(V,E)$ on $n=|V|$ vertices and $m=|E|$ edges, 
the following notations will be used throughout the paper.\\[-2mm]
\begin{itemize}
\item $V(H)$~:~ The set of vertices present in a graph $H$.
\item $E(H)$~:~ The set of edges present in a graph $H$.
%\item $H[A]$~:~ The subgraph of $H$ induced by vertices in set $A$.
%\item $G^R$~:~ The graph obtained by reversing all the edges in graph $G$.
\item $H\setminus F$~:~ For a set of edges $F$, the graph obtained by deleting the edges in $F$ from graph $H$.
%\item $\pi(x,y,H)$~:~ The shortest path from $x$ to $y$ in graph $H$.
% such a path is not arbitraly chosen, or prefixed, depends upon context. 
%\item $s\leadsto t$ path~:~ An (arbitrary) directed path from a vertex $s$ to another vertex $t$.
\item $s-t$ path~:~ A directed path from a vertex $s$ to another vertex $t$.
\item $P \circ Q$~:~ The concatenation of two paths $P$ and $Q$, i.e., a path that first follows $P$ and then $Q$.
\item $P[L]$~:~ The subpath of the path $P$ containing the first $L$ vertices of $P$.
\item $P[-L]$~:~ The subpath of the path $P$ containing the last $L$ vertices of $P$.
\item $P[u-v]$~:~ The $u-v$ subpath of the path $P$.
%\item $T(v)$~:~ The subtree of a directed tree $T$ rooted at a vertex $v\in T$.
\end{itemize}

%Let us first formally define fault-tolerant reachability subgraph ({\ftrs}) for a set of node-pairs.
%\begin{definition}[$\ftrs$]
%Let $\P\in V\times V$ be any set of pairs of vertices. A subgraph $H$ of $G$ is said to be a \emph{$k$-Fault-Tolerant Reachability Subgraph} of $G$ for $\P$, denoted as $\kftrs(G,\P)$, if for any pair $(s,t)\in \P$ and for any subset $F\subseteq E$ of at most $k$ edges, $t$ is reachable from $s$ in $G \setminus F$ if and only if $t$ is reachable from $s$ in $H \setminus F$.
%\label{definition:FTRS}
%\end{definition}

Our algorithm for computing {\em pairwise-reachability} preservers (and oracles) in a fault tolerant environment employs the concept of a {\em single-source} {\ftrs} which is a sparse subgraph that preserves reachability from a designated source vertex even after the failure of at most $k$ edges in $G$. Baswana \emph{et al.}~\cite{BCR16} provide a construction of sparse $\kftrs$ for any general $k\ge 1$ when there is a designated source vertex.
\begin{theorem}[\cite{BCR16}]
\label{theorem:ftrs}
For any directed graph $G=(V,E)$, a designated source vertex $s\in V$, and an integer $k\geq 1$, there exists a (sparse) subgraph $H$ of $G$ which is a $\kftrs(G,\{s\}\times V)$ and contains at most $2^k n$ edges. Moreover, such a subgraph is computable in $O(2^kmn)$ time, where $n$ and $m$ are respectively the number of vertices and edges in graph $G$.
\end{theorem}
For our constructions of pairwise 2-$\ftrs$ (and 2-$\ftro$), we will start with an alternate construction of single-source 2-$\ftrs$ (see Section~\ref{section:single-pair}).
 
%Observe that in case of no failure, a directed reachability tree has $n-1$ edges and is able to preserve reachability from the source which is a also the root. %Thus, a single-source FTRS is an analogous of reachability-tree in presence of failures. 
%An {\ftrs} with respect to a given source is formally defined as follows.

%Similar to the fault-tolerant rechability subgraph, below we formally define the fault-tolerant reachability oracle for a set of node-pairs.
%\begin{definition}[$\ftro$]
%Let $\P\in V\times V$ be any set of pairs of vertices. For a graph $G$, a data structure $DS(G)$ is said to be a \emph{$k$-Fault-Tolerant Reachability Oracle} of $G$ for $\P$, denoted as $\kftro(G,\P)$, if given a query with any pair $(s,t)\in \P$ and any subset $F\subseteq E$ of at most $k$ edges, $DS(G)$ decides whether $t$ is reachable from $s$ in $G \setminus F$ in time $O(1)$.
%\label{definition:FTRO}
%\end{definition}
Note, in the $\ftro$ definition we restrict ourselves to a data structure with constant query time. (The term \emph{oracle} came from its ability to answer a query in constant time.) Unlike single-source $\kftrs$, there is no non-trivial construction of single-source $\kftro$ for $k\ge 3$. For $k = 2$, the following result by Choudhary~\cite{Choudhary16} provides an $O(n)$ size oracle.

\begin{theorem}[\cite{Choudhary16}]
\label{theorem:ftro}
There is a polynomial time algorithm that given any directed graph $G=(V,E)$, a designated source vertex $s\in V$, constructs a 2-$\ftro(G,\{s\}\times V)$ of size $O(n)$.
\end{theorem}

Our constructions will require the knowledge of the vertices reachable from a vertex $s$ as well as the  vertices that can reach $s$. So we will use $\ftrs$ (and $\ftro$) defined with respect to a source vertex ($\{s\}\times V$ case), as well as $\ftrs$ (and $\ftro$) defined with respect to a destination vertex ($V\times \{s\}$ case). 

In this paper, we consider fault-tolerant structures with respect to edge failures only.
Vertex failures can be handled by simply splitting a vertex $v$ into an edge $(v_{in},v_{out})$,
where the incoming and outgoing edges of $v$ are respectively directed into $v_{in}$ and directed out
of $v_{out}$.

\paragraph*{Static Dictionary Data Structure. }One of the classic data structure problems is the \emph{dictionary problem}. In this problem, we are given a set $S$ of items, where each item is a (key,value) pair. Keys are from an arbitrary universe $U$. The task is to build a data structure that can efficiently answer the following query.

 Search($k$): For a given $k$, if the key $k$ is present in $S$, return the value associated with the key $k$ in $S$; otherwise return NULL.

The celebrated result by Fredman, Koml{\'o}s and Szemer{\'e}di~\cite{fredman1984storing} provide a data structure of size $O(|S|)$ with worst-case query (search) time $O(1)$. Since then various randomized and deterministic schemes have been developed for this problem. Currently, the best deterministic data structure attains $O(|S|)$ space bound and $O(1)$ query time. One of the key ingredients of our constructions is this data structure of the (static) dictionary problem. In particular, we use the following theorem.
\begin{theorem}[\cite{ruvzic2008uniform}]
\label{thm:dictionary}
Given any set $S$, there is a data structure of size $O(|S|)$, for the dictionary problem, that can answer each search query in worst-case $O(1)$ time. Moreover, the data structure can be constructed in deterministic polynomial (in $|S|$) time.
\end{theorem}

\paragraph*{Fractional Hitting Set. }Let $U:=\{1,2,\cdots,n\}$ be the universe. A set $S \subseteq U$ is said to be a \emph{hitting set} of a family of subsets $S_1,S_2,\cdots,S_m \subseteq U$ if and only if for all $i\in [m]$, $S\cap S_i \ne \emptyset$ (i.e., $S$ intersects all the subsets $S_i$). Given a family of subsets, the problem of finding the minimum sized hitting set is referred to the \emph{hitting set problem}, and is known to be NP-complete. However, there is a simple greedy algorithm that given any family of subsets $S_1,\cdots,S_m$ each of size at least $k$, finds a hitting set of size at most $O(\frac{n}{k} \log n)$.

In this paper, we consider a variant of the hitting set problem, which we call \emph{fractional hitting set}. A set $S \subseteq U$ is said to be a hitting set with slack of a family of subsets $S_1,\cdots,S_m$ if and only if $|\{i \in [m] \mid S\cap S_i \ne \emptyset\}| = \Omega(m)$ (i.e., $S$ has a non-empty intersection with at least a constant fraction of the input subsets). It is not difficult to see that the above mentioned greedy hitting set algorithm finds a fractional hitting set of size at most $O(n/k)$.

\begin{theorem}
\label{thm:frac-hitting-set}
There is a deterministic algorithm that given a family of subsets $S_1,\cdots,S_m \subseteq [n]$ each of size at least $k$, finds a fractional hitting set of size at most $O(n/k)$ in time $O(n + mk + \frac{n}{k} \log n)$.
\end{theorem}
\begin{proof}
W.l.o.g. assume, each $S_i$ is of size exactly $k$; otherwise we can just consider any $k$ elements from each $S_i$. Let us now consider a slightly modified version of the standard greedy algorithm for the hitting set problem. Initialize $\F:=\{S_1,\cdots,S_m\}$, and $S=\emptyset$. For each $v\in [n]$, compute the number of subsets in $\F$ that contains $v$, i.e., $c(v):=|\{S_i \in \F \mid v \in S_i\}|$. Let $v_{max}$ be a maximizer of $c(v)$. Include $v_{max}$ in $S$, i.e, $S \gets S \cup \{v_{max}\}$. Remove all the subsets that contains $v_{max}$ from $\F$, i.e., $\F \gets \F \setminus \{S_i \mid v_{max} \in S_i\}$. Continue this process $R=4n/k$ times, and output $S$. (Note, instead to continuing the process until $\F=\emptyset$ as in the standard greedy algorithm for the hitting set problem, we terminate after a fixed number of steps.)

Clearly, $|S|=R=O(n/k)$. We would like to claim that $S$ is a fractional hitting set. For that purpose, it suffices to show that at the end of the algorithm $|\F| \le m/10$. Let $n_j$ denote the size of $\F$ after the $j$-th element is added in $S$. $n_0=m$, and we need to show that $n_R \le m/10$. Consider the step when the $j$-th element has been added in $S$. Observe, $\sum_{v \in [n]\setminus H} c(v) = n_{j-1} k$, and thus by a simple averaging, an element $v_{max}^j$ that maximizes $c(v)$ must satisfy $c(v_{max}^j) \ge \frac{n_{j-1}k}{n - j +1}$. Thus,
$$n_j \le \Big(1-\frac{k}{n-j+1}\Big)n_{j-1} \le m \prod_{r=0}^{j-1}\Big(1-\frac{k}{n-r}\Big) < m(1-k/n)^j \le me^{-kj/n}.$$
Then for $R=4n/k$, $n_R < m/10$.

The running time of the modified algorithm is bounded by the standard greedy algorithm for the hitting set problem. For the sake of completeness, let us briefly comment on the running time analysis. At the beginning we can compute the values of $c(v)$'s, and then store them using any standard heap data structure (e.g.~\cite{brodal2012strict}). Then at each step, we extract the max value, and then we update the values of the others. This gives the total time to be $O(n + mk + \frac{n}{k}\log n)$.
\end{proof}

\paragraph*{Sparsifiers with Slack. }In case of pairwise FTRO (FTRS) we want to construct an oracle (subgraph) that provides correct reachability information for all the given node-pairs even after edge failures. This is the most standard definition of pairwise FTRO (FTRS). Next, we define a variant which we refer to as \emph{pairwise FTRO (FTRS) with slack}. Interested readers may find more works on sparsifiers with slack in~\cite{chan2006spanners, konjevod2007compact, dinitz2007compact, bodwin2020note}. Given a graph $G$ and a pair set $\P$, a data structure $\O$ (subgraph $H$) is a $\kftro$ with slack ($\kftrs$ with slack) if there is $\P' \subseteq \P$, $|\P'| = \Omega(|\P|)$ such that $\O$ ($H$) is a $\kftro$ ($\kftrs$) for the pair set $\P'$.

The following lemma from~\cite{bodwin2020note} show that to construct a $\kftrs$ for a pair set it suffices to construct a $\kftrs$ with slack.

\begin{lemma}[~\cite{bodwin2020note}]
\label{lem:ftrs-slack}
Let us consider constants $\alpha,\beta,\gamma >0$ and a parameter $p^*$. Consider a graph $G(V,E)$ with $n$ nodes. If
\begin{itemize}
\item There is a $\kftrs$ of size $O(n^{\alpha})$ for any node-pair set $\P$ of size at most $p^*$, and
\item There is a $\kftrs$ with slack of size $O(n^{\beta} |\P|^{\gamma})$ for any node-pair set $\P$ of size at least $p^*$,
\end{itemize}
then there is a $\kftrs(G,\P)$ of size at most $O(n^{\alpha} + n^{\beta} |\P|^{\gamma})$. Moreover, if the $\kftrs$ for at most $p^*$ pairs and the $\kftrs$ with slack for at least $p^*$ pairs can be computed in polynomial time, then the $\kftrs(G,\P)$ can also be computed in polynomial time.
\end{lemma}

A similar claim also holds for oracles. The proof is analogous to the proof of Lemma~\ref{lem:ftrs-slack} in~\cite{bodwin2020note}.

\begin{lemma}
\label{lem:ftro-slack}
Let us consider constants $\alpha,\beta,\gamma,\delta,\kappa,\nu >0$ and a parameter $p^*$. Consider a graph $G(V,E)$ with $n$ nodes. If
\begin{itemize}
\item There is a $\kftro$ of size $O(n^{\alpha})$ and query time $O(n^{\delta})$, for any node-pair set $\P$ of size at most $p^*$, and
\item There is a $\kftro$ with slack of size $O(n^{\beta} |\P|^{\gamma})$ and query time $O(n^{\kappa}|\P|^{\nu})$, for any node-pair set $\P$ of size at least $p^*$,
\end{itemize}
then there is a $\kftro(G,\P)$ of size at most $O(n^{\alpha} + n^{\beta} |\P|^{\gamma} +|\P|)$ and query time $O(\max\{n^{\delta}, n^{\kappa}|\P|^{\nu}\})$. Moreover, if the $\kftro$ for at most $p^*$ pairs and the $\kftro$ with slack for at least $p^*$ pairs can be constructed in polynomial time, then the $\kftro(G,\P)$ can also be constructed in polynomial time.
\end{lemma}
Note, an extra additive term of $|\P|$ in the size bound of $\kftro$ is because the final $\kftro$ will consist of a collection of data structures each for a specific subset of pairs of $\P$. So to make the query faster, we need to maintain a pointer for each pair in $\P$ to specify which data structure to query.

\section{An (Alternate) Construction of 2-{\ftrs} for a Single Pair}
\label{section:single-pair}
Given any directed graph $G$ with $n$ vertices, for any pair of vertices $(s,t)$, we already get an $O(n)$ size 2-$\ftrs(G,(s,t))$ from~\cite{BCR16} (see Theorem~\ref{theorem:ftrs}). However, now we study the structure of a 2-$\ftrs(G,(s,t))$ more closely and provide a slightly different construction which will be useful for our constructions of 2-$\ftro$ and 2-$\ftrs$ for multiple pairs. Let us first define a notation to denote the ordering of vertices on a specific path. For any path $P$ and two vertices $u,v$ on $P$, we will write $u <_P v$ if and only if $u$ appears before $v$ on $P$.

Suppose we are given a directed graph $G$ and a pair $(s,t)$. Let us consider two $s-t$ paths $P_{(s,t)}^1$ and $P_{(s,t)}^2$, that intersect only at $s-t$ cut-edges and $s-t$ cut-vertices. We call these two paths \emph{outer strands}. Next, we define coupling paths between these two outer strands. For each vertex on the outer strands, let us first define two \emph{coupling points} of it (one on each outer strand). Consider a vertex $v \in V(P_{(s,t)}^1 \cup P_{(s,t)}^2)$ and $i \in \{1,2\}$. Then consider the vertex $u_{(s,t),v}^i\in P_{(s,t)}^i$ (if exists) for which there is a $u_{(s,t),v}^i-v$ path that is edge-disjoint with both the outer strands, and there is no vertex $u'\in P_{(s,t)}^i$ such that $u' <_{P_{(s,t)}^i} u_{(s,t),v}^i$ and there is a $u'-v$ path that is edge-disjoint with both the outer strands. We refer to $u_{(s,t),v}^1, u_{(s,t),v}^2$ as \emph{coupling points} of $v$. For each $v$ on the outer strands and each coupling point $u_{(s,t),v}^i$ of it, consider any arbitrary $u_{(s,t),v}^i-v$ path (denoted as $Q_{(s,t),v}^i$) of $G$, that is edge-disjoint with both the outer strands. We refer to the path $Q_{(s,t),v}^i$ as a \emph{coupling path} to $v$. Note, for each vertex $v$ on the outer strands, there are at most two coupling points of it (one on each outer strand) and from each coupling point there is a coupling path to $v$. We take two outer strands and some of the coupling paths to build a subgraph $H_{(s,t)}$. More specifically, for a vertex $v \in V(P_{(s,t)}^j)$, if for all the vertices below $v$ on that strand, their coupling points on a strand $P_{(s,t)}^i$ is below the coupling point of $v$ on that strand, then only add the coupling path from the strand $P_{(s,t)}^i$ to $v$. Later, we prove that keeping these "essential" coupling paths is sufficient. 

Below we provide a formal description of the procedure to construct $H_{(s,t)}$. Given $G$ and a pair $(s,t)$, we build a subgraph $H_{(s,t)}$ as follows:
\begin{enumerate}
\item $V(H_{(s,t)})=V(G)$.
\item Add all the edges on the paths $P_{(s,t)}^1$, $P_{(s,t)}^2$ (two outer strands) to $H_{(s,t)}$.
\item For all $i,j \in \{1,2\}$ and $v \in V(P_{(s,t)}^j)$
\begin{enumerate}
\item If for all $v' \in V(P_{(s,t)}^j[v-t])$, $u_{(s,t),v}^i <_{P_{(s,t)}^i} u_{(s,t),v'}^i$, then take the coupling path $Q_{(s,t),v}^i$, and add all its edges to $H_{(s,t)}$.
\end{enumerate}
\end{enumerate}
Readers may refer to Figure~\ref{fig:single-pair-ftrs} for an example subgraph $H_{(s,t)}$. Note, the above construction procedure of the subgraph $H_{(s,t)}$ runs in polynomial time. 

Before arguing that $H_{(s,t)}$ is a 2-$\ftrs(G,(s,t))$, let us introduce a notion of \emph{nice path}. We call an $s-t$ path in $H_{(s,t)}$ \emph{nice} if it first follows a outer strand till some vertex $u$, then take a coupling path to some vertex $u'$, and finally follows the outer strand on which $u'$ lies till $t$. More formally, an $s-t$ path in $H_{(s,t)}$ is nice if and only if it is of the form $P_{(s,t)}^i[s-u] \circ Q_{(s,t),u'}^{i} \circ P_{(s,t)}^j[u'-t]$, for some $i,j \in \{1,2\}$, where subpath $Q_{(s,t),u'}^{i}$ starts from $u$ and is edge-disjoint with both the outer strands $P_{(s,t)}^1$, $P_{(s,t)}^2$. ($Q_{(s,t),u'}^{i}$ could be empty.)

The following claim shows that $H_{(s,t)}$ is a 2-$\ftrs(G,(s,t))$, and also will be useful in the subsequent sections while arguing about the correctness of our pairwise 2-{\ftro} and 2-{\ftrs} construction.
\begin{claim}
\label{clm:special-path}
For any two edges $f_1,f_2$, if there is an $s-t$ path in $G \setminus \{f_1,f_2\}$, then there must be a nice $s-t$ path in $H_{(s,t)} \setminus \{f_1,f_2\}$.
\end{claim}
\begin{proof}
First of all, one of $f_1,f_2$ must be on $P_{(s,t)}^1$ and another on $P_{(s,t)}^2$. Otherwise, since $P_{(s,t)}^1$ and $P_{(s,t)}^2$ intersect only at $s-t$ cut-edges and $s-t$ cut-vertices, at least one of $P_{(s,t)}^1$, $P_{(s,t)}^2$ remains intact in $H_{(s,t)} \setminus \{f_1,f_2\}$, and thus there is nothing to prove. So from now on without loss of generality, assume, $f_1=(x_1,y_1)$ lies on $P_{(s,t)}^1$ and $f_2=(x_2,y_2)$ lies on $P_{(s,t)}^2$. 

Let $P$ be an $s-t$ path in $G \setminus \{f_1,f_2\}$. Observe, $P$ must intersect either $P_{(s,t)}^1[y_1-t]$ or $P_{(s,t)}^2[y_2-t]$. Let $z'$ be the first vertex on $P$, that is also in $V(P_{(s,t)}^1[y_1-t] \cup P_{(s,t)}^2[y_2-t])$ (such a vertex exists because $P$, $P_{(s,t)}^1$ and $P_{(s,t)}^2$ all end at $t$). Let $z'$ lie on $P_{(s,t)}^j$, for some $j\in \{1,2\}$. Now consider the last vertex before $z'$ on $P$, that is also in $V(P_{(s,t)}^1 \cup P_{(s,t)}^2)$ (such a vertex exists because $P$, $P_{(s,t)}^1$ and $P_{(s,t)}^2$ all start from $s$), and let $z$ be that vertex which is on $P_{(s,t)}^i$, for some $i \in \{1,2\}$. ($i$ could be equal to $j$.) Observe, $z$ must lie on the subpath $P_{(s,t)}^i[s-x_i]$. Furthermore, $P[z-z']$ is edge-disjoint with $P_{(s,t)}^1$ and $P_{(s,t)}^2$. Hence, by the construction of $H_{(s,t)}$, there must be a path $Q_{(s,t),u'}^i$ for some $u'\in V(P_{(s,t)}^j[z'-t])$ which starts from a vertex $u<_{P_{(s,t)}^i} z<_{P_{(s,t)}^i} x_i$ ($u$ could be equal to $z$). Hence, $P_{(s,t)}^i[s-u]$ remains intact in $H_{(s,t)} \setminus \{f_1,f_2\}$. Note, $u' \in V(P_{(s,t)}^j[y_j-t]$ (because $u'\in V(P_{(s,t)}^j[z'-t])$). Thus the subpath $P_{(s,t)}^j[u'-t]$ also remains intact in $H_{(s,t)} \setminus \{f_1,f_2\}$. Lastly, since $Q_{(s,t),u'}^i$ is edge-disjoint with $P_{(s,t)}^1$ and $P_{(s,t)}^2$, it also remains intact in $H_{(s,t)} \setminus \{f_1,f_2\}$. Hence, we get an $s-t$ path $P_{(s,t)}^i[s-u] \circ Q_{(s,t),u'}^{i} \circ P_{(s,t)}^j[u'-t]$ in $H_{(s,t)} \setminus \{f_1,f_2\}$.
\end{proof}

Next, let us use the above claim to conclude that $H_{(s,t)}$ is a 2-$\ftrs(G,(s,t))$. Furthermore, we also argue that the size of $H_{(s,t)}$ is $O(n)$.

\begin{lemma}
\label{lem:single-pair-ftrs}
$H_{(s,t)}$ is a 2-$\ftrs(G,(s,t))$ and contains $O(n)$ edges.
\end{lemma} 
\begin{proof}
Since $H_{(s,t)}$ is a subgraph of $G$, from Claim~\ref{clm:special-path}, we conclude that for any two edges $f_1,f_2$, there is an $s-t$ path in $G \setminus \{f_1,f_2\}$ if and only if there is an $s-t$ path in $H_{(s,t)} \setminus \{f_1,f_2\}$. So, $H_{(s,t)}$ is a 2-$\ftrs(G,(s,t))$.

Let us now compute the size of $H_{(s,t)}$. First, we would like to claim that for each vertex $v \in V(H_{(s,t)}) \setminus V(P_{(s,t)}^1 \cup P_{(s,t)}^2)$, in-degree of $v$ in $H_{(s,t)}$ is at most 4. If not, then there exists a vertex $v \in V(H_{(s,t)}) \setminus V(P_{(s,t)}^1 \cup P_{(s,t)}^2)$ with at least five incoming edges $e_1,e_2,e_3,e_4,e_5$. Suppose during the construction of $H_{(s,t)}$, for each $r \in [5]$, $e_r$ has been included while adding a path $Q_{(s,t),v_r}^{i_r}$, where $v_r \in V(P_{(s,t)}^{j_r})$. Since for each $r \in [5]$, $i_r,j_r \in \{1,2\}$, by the pigeonhole principle, there must exist $r\ne r'\in [5]$ such that $i_r=i_{r'}(=i)$ and $j_r=j_{r'}(=j)$. Without loss of generality, assume, $v_r <_{P_{s,t}^{j}} v_{r'}$. Since $Q_{(s,t),v_r}^{i}$ and $Q_{(s,t),v_{r'}}^{i}$ intersect at $v$, either $u_{(s,t),v_{r'}}^{i}=u_{(s,t),v_r}^{i}$ or $u_{(s,t),v_{r'}}^{i} <_{P_{(s,t)}^i} u_{(s,t),v_r}^{i}$. Hence, by the Step 3(a) of the construction procedure of $H_{(s,t)}$, $Q_{(s,t),v_r}^{i_r}$ would not be added, which leads to a contradiction.

Observe, by the construction of $H_{(s,t)}$, in-degree of each $v \in V(P_{(s,t)}^1 \cup P_{(s,t)}^2)$ is also at most 4. Hence, $H_{(s,t)}$ has total $O(n)$ edges.
\end{proof}

\section{Dual Fault-tolerant Pairwise Reachability Oracle}
\label{section:dual-ftro}
In this section, we present a construction of a dual fault-tolerant reachability oracle (2-$\ftro$) for a given $n$ vertex (directed) graph $G=(V,E)$ and a vertex-pair set $\P \subseteq V \times V$. The size of the data structure/oracle is $O(n\sqrt{|\P|})$. Upon given any two failure edges $f_1,f_2$ and a pair $(s,t) \in \P$, the oracle correctly decides whether there is an $s-t$ path in $G\setminus \{f_1,f_2\}$ in $O(1)$ time. In particular, we prove Theorem~\ref{thm:dual-oracle}.

Here we provide a construction of 2-$\ftro$ with slack, which together with Lemma~\ref{lem:ftro-slack} shows Theorem~\ref{thm:dual-oracle}. 
\begin{theorem}
\label{thm:dual-oracle-slack}
A directed graph $G=(V,E)$ with $n$ vertices can be processed in randomized polynomial time for a given set $\P \subseteq V \times V$ of vertex-pairs, to build a 2-$\ftro$ with slack of size $O(n \sqrt{|\P|})$, having query time $O(1)$.
\end{theorem}

Let us start with the description of the oracle with slack. Below we will use the notations $H_{(s,t)}, P_{(s,t)}^i, Q_{(s,t),v}^j$ defined in Section~\ref{section:single-pair}.\\
\textbf{Procedure to Construct a 2-$\ftro$ with slack for the input $(G,\P)$:}
\begin{enumerate}
\item First, for all the pairs $(s,t) \in \P$, construct $H_{(s,t)}=2\text{-}\ftrs(G,(s,t))$ using the construction provided in Section~\ref{section:single-pair}. (Recall, $H_{(s,t)}$ contains two $s-t$ paths $P^1_{(s,t)}$ and $P^2_{(s,t)}$, that intersect only at $s-t$ cut-edges and $s-t$ cut-vertices.)
\item Consider the following family of subsets $\F:=\{V(P^i_{(s,t)}[L]), V(P^i_{(s,t)}[-L]) \mid i\in \{1,2\},(s,t) \in \P\}$, where $L=\Theta(n/\sqrt{|\P|})$. Construct a fractional hitting set $S$ of $\F$ (by Theorem~\ref{thm:frac-hitting-set}). Let $\Q:=\{(s,t)\in \P \mid \forall_{i\in \{1,2\}} V(P^i_{(s,t)}[L])\cap S \ne \emptyset \text{ and } V(P^i_{(s,t)}[-L]) \cap S \ne \emptyset\}$.
\item For each $v \in S$, construct a 2-$\ftro(G,\{v\} \times V)$ $O_v$ and a 2-$\ftro(G,V \times \{v\})$ $I_v$ (using Theorem~\ref{theorem:ftro}), and store them.
\item For each pair $(s,t) \in \Q$, maintain a list $T_{(s,t)}$ of size 4, that stores an (arbitrarily chosen) vertex from the set $P^i_{(s,t)}[L]\cap S$ and an (arbitrarily chosen) vertex from the set $P^i_{(s,t)}[-L]\cap S$ for each $i \in \{1,2\}$.
\item For each pair $(s,t) \in \Q$, let $M_{(s,t)}=V\big(\bigcup_{i \in \{1,2\}}(P_{(s,t)}^i[L]\cup P_{(s,t)}^i[-L])\big)$. Build a dictionary data structure $\mathcal{M}_{(s,t)}$ for the set $M_{(s,t)}$ (using Theorem~\ref{thm:dictionary}).
\item For each pair $(s,t) \in \Q$, construct the following auxiliary graph $A_{(s,t)}$ (see Figure~\ref{fig:auxiliary}):
\begin{enumerate}
\item Vertex set of $A_{(s,t)}$ is $V\big(\bigcup_{i \in \{1,2\}}(P_{(s,t)}^i[L]\cup P_{(s,t)}^i[-L])\big)$.
\item For each edge in $E\big(\bigcup_{i \in \{1,2\}}(P_{(s,t)}^i[L]\cup P_{(s,t)}^i[-L])\big)$, add an edge to $A_{(s,t)}$. We refer these edges as \emph{path edges}.
\item For all pair of vertices $u,v \in V(A_{(s,t)})$, if there is a $u-v$ path in $H_{(s,t)}$, that is edge-disjoint with $P^1_{(s,t)}$ and $P^2_{(s,t)}$, then add an edge $(u,v)$ to $A_{(s,t)}$. We refer these edges as \emph{auxiliary edges}.
\end{enumerate}
\item For each pair $(s,t) \in \Q$, construct a 2-$\ftro(A_{(s,t)},s)$ $D_{(s,t)}$.
\item The final data structure consists of $I_v, O_v$ for each $v \in S$, and for each pair $(s,t) \in \Q$, the list $T_{(s,t)}$, the dictionary data structure $\mathcal{M}_{(s,t)}$, and $D_{(s,t)}$.
\end{enumerate}

We will show that the data structure constructed using the above procedure is a 2-$\ftro(G,\Q)$. Given any two edges $f_1,f_2$ and a pair $(s,t) \in \Q$, the query algorithm works as follows:\\
\textbf{Query Algorithm ($f_1,f_2,(s,t)$):}
\begin{enumerate}
\item For each vertex $v \in T_{(s,t)}$, check if there is an $s-v$ and $v-t$ path in $G\setminus\{f_1,f_2\}$ using $I_v$ and $O_v$ respectively. If for at least one vertex $v\in T_{(s,t)}$, both $s-v$ and $v-t$ path exist, then return "there is an $s-t$ path in $G\setminus \{f_1,f_2\}$".
\item Otherwise, check whether any of the endpoints of $f_1$ and $f_2$ is not present in $\mathcal{M}_{(s,t)}$. If any of them is not, discard that edge and continue to Step 3. Let $F \subseteq \{f_1,f_2\}$ be the set of remaining edges after the above pruning.
\item Use the oracle $D_{(s,t)}$ to decide whether there is an $s-t$ path in $A_{(s,t)} \setminus \{f_1,f_2\}$. If yes, then return "there is an $s-t$ path in $G\setminus \{f_1,f_2\}$"; else return "there is no $s-t$ path in $G\setminus \{f_1,f_2\}$".
\end{enumerate}

Note, the sole purpose of step 2 of the above query algorithm is to make sure the edges $f_1, f_2$ are defined over the set of vertices present in the auxiliary graph $A_{(s,t)}$. If any of them is not, then there is no point in deleting that edge from $A_{(s,t)}$ (at step 3). 

\paragraph*{Size of the data structure. }By Theorem~\ref{thm:frac-hitting-set}, $|S|=O(\sqrt{|\P|})$. For a vertex $v \in S$, to store $I_v, O_v$ we need $O(n)$ space by Theorem~\ref{theorem:ftro}. So, for all the $I_v$'s and $O_v$'s, we need $O(n\sqrt{|\P|})$ space. To store the lists $T_{(s,t)}$, for all $(s,t) \in \Q$, we need $O(|\Q|)$ space. Since for a pair $(s,t) \in \Q$, the set $M_{(s,t)}$ is of size at most $O(L)$, the dictionary data structure $\mathcal{M}_{(s,t)}$ requires $O(L)$ space using Theorem~\ref{thm:dictionary}. Next, observe, for each $(s,t) \in \Q$, the auxiliary graph $A_{(s,t)}$ contains at most $O(L)$ vertices and thus at most $O(L^2)$ edges. However, we do not store $A_{(s,t)}$ as part of the data structure. Instead, we just store $D_{(s,t)}=2\text{-}\ftro(A_{(s,t)},s)$, which requires $O(L)$ space again by Theorem~\ref{theorem:ftro}. So, the total space required by the whole data structure is $O(n\sqrt{|\P|} + L\cdot |\Q|)=O(n \sqrt{|\P|})$, for $L=\Theta(n/\sqrt{|\P|})$ (note, $|\Q| \le |\P|\le n^2$).

\paragraph*{Query time. }By Theorem~\ref{theorem:ftro}, each of $I_v, O_v$ and $D_{(s,t)}$ has query time $O(1)$. Checking the presence of the end points of $f_1,f_2$ in the data structure $\mathcal{M}_{(s,t)}$ also requires $O(1)$ time by Theorem~\ref{thm:dictionary}. Thus the total query time is also $O(1)$.

%We defer the correctness of the query algorithm to Appendix~\ref{app:correctness-ftro}.

\paragraph*{Correctness of pairwise 2-{\ftro} with slack. } Proving the constructed data structure is a 2-$\ftro$ with slack consists of two steps. First, we show that the constructed data structure is a 2-$\ftro(G,\Q)$. Second, we show that $|\Q| = \Omega(|\P|)$. The second part is quite straightforward. By the definition of fractional hitting set, $S$ intersects all but $1/10$ fraction of the subsets of $\F$. Hence, $|\Q| \ge \frac{3}{5} |\P|$. So it only remains to show the first part, for which we need to show that the query algorithm is correct for any pair in $\Q$.

\iffalse
Consider the following two good events: $\mathcal{E}$: For all $(s,t)\in \Q$, no failure occurs during the construction of $\mathcal{M}_{(s,t)}$. By Theorem~\ref{thm:dictionary} together with a simple union bound, the event $\mathcal{E}$ occurs with probability at least $1-1/n$. From now on assume the above event occurs. (Note, during the construction of the oracle, it is possible to check whether the event occurs, in polynomial time. So we can even declare "failure" at the end of the construction procedure if $\mathcal{E}$ does not occur.)
  \fi
  
  Now, to prove the correctness of the query algorithm, it suffices to show the following.
\begin{lemma}
\label{lem:dual-oracle-correctness}
For any pair $(s,t) \in \Q$ and two edges $f_1,f_2$, there is an $s-t$ path in $G\setminus \{f_1,f_2\}$ if and only if
\begin{enumerate}
\item[(a)] Either $\exists v \in T_{(s,t)}$ such that there is an $s-v$ and $v-t$ path in $G\setminus \{f_1,f_2\}$,
\item[(b)] Or, there is an $s-t$ path in $A_{(s,t)} \setminus \{f_1,f_2\}$.
\end{enumerate}
\end{lemma}

Note, $f_1$ or $f_2$ may not lie on the auxiliary graph $A_{(s,t)}$. In case an edge $f$ does not belong to $A_{(s,t)}$, $A_{(s,t)} \setminus \{f\}=A_{(s,t)}$. Before proving Lemma~\ref{lem:dual-oracle-correctness}, let us now argue on the correctness of the query algorithm by assuming this lemma. Observe, step 1 of the query algorithm checks whether $\exists v \in T_{(s,t)}$ such that there is an $s-v$ and $v-t$ path in $G\setminus \{f_1,f_2\}$. Step 2 prunes the edges from $\{f_1,f_2\}$ if some of them are not defined over the vertices of $A_{(s,t)}$, and let $F \subseteq \{f_1,f_2\}$ be the remaining set of edges after the pruning. Then in step 3, we decide whether there is an $s-t$ path in $A_{(s,t)} \setminus F$. It is thus immediate that the correctness of the query algorithm directly follows from Lemma~\ref{lem:dual-oracle-correctness}.

\begin{proof}[Proof of Lemma~\ref{lem:dual-oracle-correctness}]
Consider the subgraph $H_{(s,t)}$, which is a 2-$\ftrs(G,(s,t))$ (that we have constructed in the step 1 of the oracle construction procedure).

\textbf{($\Rightarrow$)} Suppose there is an $s-t$ path in $G \setminus \{f_1,f_2\}$. By Claim~\ref{clm:special-path}, there must be a nice $s-t$ path $R$ in $H_{(s,t)} \setminus \{f_1,f_2\}$. By definition of nice path, $R$ would be of the form $P_{(s,t)}^i[s-u] \circ Q_{(s,t),u'}^{i} \circ P_{(s,t)}^j[u'-t]$, for some $i,j \in \{1,2\}$, where the subpath $Q_{(s,t),u'}^{i}$ starts from the vertex $u$ and is edge-disjoint with $P_{(s,t)}^1$, $P_{(s,t)}^2$. ($Q_{(s,t),u'}^{i}$ could be empty; in that case, $R=P_{(s,t)}^i$ for some $i \in \{1,2\}$.) By the definition of $\Q$, $P_{(s,t)}^i[L]\cap S \ne \emptyset$ and $P_{(s,t)}^j[-L]\cap S \ne \emptyset$. Suppose $v \in P_{(s,t)}^i[L]\cap S$ and $v' \in P_{(s,t)}^j[-L]\cap S$ are stored in $T_{(s,t)}$. Now, if $u$ does not lie on $P_{(s,t)}^i[L]$, i.e., the length of the subpath $P_{(s,t)}^i[s-u]$ is at least $L$, then $R$ must pass through $v \in P_{(s,t)}^i[L] \cap S$. Similarly, if $u'$ does not lie on $P_{(s,t)}^j[-L]$, then $R$ must pass through $v' \in P_{(s,t)}^j[-L] \cap S$. In both the scenarios, for some $v \in T_{(s,t)}$, there is an $s-v$ and $v-t$ path in $H_{(s,t)} \setminus \{f_1,f_2\}$ (and thus in $G \setminus \{f_1,f_2\}$).

So let us now focus on the case when $u$ lies on $P_{(s,t)}^i[L]$ and $u'$ lies on $P_{(s,t)}^j[-L]$. Then by the construction of the auxiliary graph $A_{(s,t)}$, $u,u' \in V(A_{(s,t)})$. Moreover, since $Q_{(s,t),u'}^{i}$ is edge-disjoint with $P_{(s,t)}^1,P_{(s,t)}^2$, the edge $(u,u')$ is present in $A_{(s,t)}$. Hence, the path $R_A=P_{(s,t)}^i[s-u] \circ (u,u') \circ P_{(s,t)}^j[u'-t]$ is an $s-t$ path in $A_{(s,t)}$. Observe, $f_1,f_2 \not \in P_{(s,t)}^i[s-u] \cup P_{(s,t)}^j[u'-t]$. So, $R_A$ survives in $A_{(s,t)} \setminus \{f_1,f_2\}$ if and only if $f_1$ (or $f_2$) is not the edge $(u,u')$. Note, it is possible to have $f_1$ (or $f_2$) same as the edge $(u,u')$ and still the $u-u'$ subpath $Q_{(s,t),u'}^{i}$ exists in $H_{(s,t)}\setminus \{f_1,f_2\}$. (This is because in $G$, and also in $H_{(s,t)}$, there can be an edge $(u,u')$ and also another path $Q_{(s,t),u'}^{i}$ that is edge-disjoint with $P_{(s,t)}^1, P_{(s,t)}^2$.) Recall, $P_{(s,t)}^1$, $P_{(s,t)}^2$ intersect only at the $s-t$ cut-edges and $s-t$ cut-vertices. So, if $f_1$ (or $f_2$) is the edge $(u,u')$, since $(u,u') \not \in P_{(s,t)}^1 \cup P_{(s,t)}^2$, either $P_{(s,t)}^1$ or $P_{(s,t)}^2$ survives in $H_{(s,t)} \setminus \{f_1,f_2\}$. Let $P_{(s,t)}^i$ be a surviving path. If that path is of length more than $L$, then the Case~(a) happens. Otherwise, $P_{(s,t)}^i$ is included in the auxiliary graph $A_{(s,t)}$ (by the construction) and thus also part of $A_{(s,t)}\setminus \{f_1,f_2\}$.

\textbf{($\Leftarrow$)} If $\exists v \in T_{(s,t)}$ such that there is an $s-v$ path $R_1$ and $v-t$ path $R_2$ in $G \setminus \{f_1,f_2\}$, then clearly the path $R_1 \circ R_2$ is an $s-t$ path in $G\setminus \{f_1,f_2\}$. %Since $H_{(s,t)}$ is a 2-$\ftrs(G,(s,t))$, there is also an $s-t$ path in $H_{(s,t)} \setminus \{f_1,f_2\}$.

Now, suppose $\nexists v \in T_{(s,t)}$ such that there is an $s-v$ and $v-t$ path in $G \setminus \{f_1,f_2\}$, but there is an $s-t$ path $R'$ in $A_{(s,t)}\setminus \{f_1,f_2\}$. Then we have to show that there must also be an $s-t$ path in $H_{(s,t)}\setminus \{f_1,f_2\}$. Let us denote the path $R'$ as $R'_1\circ (u_1,v_1) \circ R'_2 \circ (u_2,v_2)\circ \cdots \circ R'_\ell \circ (u_\ell,v_\ell)\circ R'_{\ell+1}$, where each $R'_j$ consists of only the path edges (that means, is a subpath of either $P_{(s,t)}^1$ or $P_{(s,t)}^2$), and $(u_j,v_j)$ is some auxiliary edge of $A_{(s,t)}$. (Any of the $R_j$'s could be empty.) Now, if both $f_1$ and $f_2$ lie on $P_{(s,t)}^1 \cup P_{(s,t)}^2$, then from $R'$ we deduce a (directed) walk from $s$ to $t$ in $H_{(s,t)}\setminus \{f_1,f_2\}$ as follows: $R'_1\circ R_{(u_1,v_1)} \circ R'_2 \circ R_{(u_2,v_2)}\circ \cdots \circ R'_\ell \circ R_{(u_\ell,v_\ell)}\circ R'_{\ell+1}$, where $R_{(u_j,v_j)}$ denote a $u_j-v_j$ path in $H_{(s,t)}$ that is edge-disjoint with $P_{(s,t)}^1$ and $P_{(s,t)}^2$. Such a path $R_{(u_j,v_j)}$ exists; otherwise there would not be an auxiliary edge $(u_j,v_j)$ in $A_{(s,t)}$ (by the construction). From the above directed walk, we can easily get an $s-t$ path $R$ using a standard process of removing the walk-segment between two repeated occurrences of the same vertex.

Clearly, $R$ is a valid $s-t$ path in $H_{(s,t)}\setminus \{f_1,f_2\}$ when all the subpaths $R_{(u_j,v_j)}$ remain intact in $H_{(s,t)}\setminus \{f_1,f_2\}$. Suppose, $R_{(u_j,v_j)}$, for some $j\in [\ell]$, contains $f_1$ (or $f_2$) and thus gets destroyed in $H_{(s,t)}\setminus \{f_1,f_2\}$. In that case, we claim that either $P_{(s,t)}^1$ or $P_{(s,t)}^2$ remains intact. The reason is as follows: First of all, if $(a,b)$ is an $s-t$ cut-edge, then there is no path from a vertex $u \in P_{(s,t)}^i[s-a]$ to a vertex $u' \in P_{(s,t)}^j[b-t]$ for any $i,j \in \{1,2\}$, that is edge-disjoint with $P_{(s,t)}^1$ and $P_{(s,t)}^2$, and thus there is no auxiliary edge $(u,u')$ in $A_{(s,t)}$. This implies, if there is an $s-t$ cut edge $e \not \in \bigcup_{i\in \{1,2\}}(P_{(s,t)}^i[L] \cup P_{(s,t)}^i[-L])$, there cannot be any $s-t$ path in $A_{(s,t)}$. Further, all the $s-t$ cut-edges must be present in any $s-t$ path in $A_{(s,t)}\setminus \{f_1,f_2\}$. Hence, $R'$ is an $s-t$ path in $A_{(s,t)}\setminus \{f_1,f_2\}$ means none of $f_1$ and $f_2$ is an $s-t$ cut-edge. Recall, $P_{(s,t)}^1$ and $P_{(s,t)}^2$ only share the $s-t$ cut-edges and $s-t$ cut-vertices. Now, if $f_1$ (or $f_2$) lies on $R_{(u_j,v_j)}$, for some $j\in [\ell]$, since $R_{(u_j,v_j)}$ is edge-disjoint with $P_{(s,t)}^1$ and $P_{(s,t)}^2$, either $P_{(s,t)}^1$ or $P_{(s,t)}^2$ survives in $H_{(s,t)} \setminus \{f_1,f_2\}$ (and so in $G\setminus\{f_1,f_2\}$ as well). This completes the proof.
\end{proof}

Let us now combine Theorem~\ref{thm:dual-oracle-slack} with Lemma~\ref{lem:ftro-slack} to finish the proof of Theorem~\ref{thm:dual-oracle}.

\begin{proof}[Proof of Theorem~\ref{thm:dual-oracle}]
For single pair, by~\cite{Choudhary16}, we have a 2-$\ftro$ of size $O(n)$ and query time $O(1)$. By Theorem~\ref{thm:dual-oracle-slack}, we have a 2-$\ftro$ with slack of size $O(n\sqrt{|\P|})$ and query time $O(1)$. Theorem now directly follows from Lemma~\ref{lem:ftro-slack}.
\end{proof}

%\section{Reachability Oracle under Single Failure}

\section{Single Fault-tolerant Pairwise Reachability Oracle}
\label{section:single-ftro}

In this section we will focus on computing a compact pairwise reachability oracle for single failure. We first we discuss about the single vertex failure case and then extend that result to a single edge failure. For a graph $G=(V,E)$ and a vertex-set $S \subseteq V$, we use the notation $G - S$ to denote the graph obtained from $G$ by deleting all the vertices in $S$ (and all the edges incident on them). We prove the following theorem for single failure.

\begin{theorem}
Let $G=(V,E)$ be a digraph with $n$ vertices and $\P\subseteq V\times V$ be a set of demand pairs.
We can preprocess $G$ in polynomial time to construct an $O(n+|\P| \sqrt{n})$-sized oracle
that given any failing edge $f$ (or vertex~$x$), and any pair $(s,t)\in \P$, reports 
whether or not there is an $s$ to $t$ path in $G \setminus \{f\}$ (or $G - x$) in $O(1)$ time.
\label{theorem:1ftro}
\end{theorem}
This theorem together with Theorem~\ref{thm:dual-oracle} implies Theorem~\ref{thm:single-oracle}. The above theorem is especially interesting as we get a {\ftro} of linear (in number of vertices) size
whenever number of pairs for which we have to preserve reachability is at most $O(\sqrt n)$.
%To complement this bound, prove in Section 9 that for $\omega(\sqrt n)$ demand pairs there 
%cannot exists a linear sized  $\ftro$.

\subsection{All-Pairs reachability restricted to an (s,t)-cut-set}
\label{sec:all-pair-cut}
%Let $G$ be a directed graph and $(s,t)$ be a vertex pair in $G$. Further, let $C$ be any subset of $(s,t)$ cut-vertices. 
%We prove in this subsection an $O(|C|)$ size data-structure
%all-pairs reachability as long as the query pair and the failing vertex are restricted to set $C$.

%We present in this subsection a more general theorem (that maybe of independent interest).
One of the key ingredients of our construction of reachability oracle is the following theorem.
\begin{theorem}
Let $G$ be a digraph and $p=(s,t)$ be a vertex pair in $G$. Further, let $C$ be any subset of $(s,t)$ cut-vertices. 
Then, we can compute in polynomial time an $O(|C|)$ size oracle that for any three vertices $x,y,z\in C$ 
reports whether or not there is a $y$ to $z$ path in graph $G-x$ in $O(1)$ time.
\label{theorem:apsp_cut_set}
\end{theorem}

Let $\sigma=(w_1,w_2,\ldots,w_{|C|})$ be the cut-vertices in $C$ in the order they appear on each $s$ to $t$ path in $G$,
and $\pi$ be an arbitrary $s$ to $t$ path in $G$.
For simplicity, we assume $s,t\in C$. We use the notation $<_{\sigma}$ (and $\le_{\sigma}$) to denote the relative ordering of any two elements in the ordered set $\sigma$. For each $v\in C$, let $L(v)$ be the set of predecessors of $v$ in $\sigma$, and
 $R(v)$ be the set of successors of $v$ in $\sigma$.
 Note that a node itself is not included in its predecessor and successor set.
We borrow the ideas of loop-nesting-forests~\cite{Tarjan76} to compute two auxiliary rooted forests
$T_{pred},T_{succ}$ with vertex-set $C$ as follows.
%We borrow the ideas of loop-nesting-forests defined by Tarjan with respect to DFS trees~\cite{Tarjan76} 
%and their applications in fault-tolerant strong-connectivity~\cite{GIP17}. Over the vertex-set $C$
% compute two auxiliary rooted forests $T_{pred},T_{succ}$  as follows.

In forest $T_{pred}$, the parent of $w\in C$ is the immediate predecessor (if it exists) of $w$ in $\sigma$, say $u$, that satisfies $u,w$ 
are strongly connected in graph $G-L(u)$.
In forest $T_{succ}$, the parent of $w\in C$ is the immediate successor (if it exists) of $w$ in $\sigma$, say $u$, that satisfies $u,w$ 
are strongly connected in graph $G-R(u)$. 
We compute the Lowest Common Ancestor (LCA)
and Level Ancestor (LA) data-structure over the two trees that takes $O(|C|)$
space and answers any LCA/LA query in constant time.

\begin{lemma}
For any $x,y\in C$ satisfying $x<_{\sigma}y$, we can find the immediate successor of $x$
strongly-connected to $y$ in $G-x$ in $O(1)$ time.
\label{lemma:find_successor}
\end{lemma}

\begin{proof}
Let $r$ be root of tree in $T_{pred}$ containing $y$. If $x<_\sigma r$, then we output $r$ since no 
predecessors of $r$ are strongly connected with $y$. So let us suppose $r\leqslant_\sigma x$. 
In such a case $x,y,r$ are contained in same tree in $T_{pred}$ (with root $r$),
since $x,y,r$ are strongly-connected in $G$.

Let $w$ to be the immediate successor of $x$ satisfying $y,w$ are strongly-connected in $G-x$.
Then $x<_\sigma w\leq_\sigma y$. We will prove  $w$ is also the immediate successor of $x$ satisfying 
$y,w$ are strongly-connected in $G-L(w)$. Each $y -  w$ path is disjoint with set $L(w)$. 
Indeed, if a $y$ to $w$ path contains a vertex in $L(w)\cap R(x)$, then it violates the definition of $w$;
and if a $y$ to $w$ path contains a vertex in $L(x)$ then it violates the fact that any path from $L(x)$ to $R(x)$
must pass through $x$. Now $\pi[w,y]$ is a $w$ to $y$ path disjoint with $L(w)$. 
Hence, $y,w$ are strongly-connected in $G-L(w)$.
Clearly, for any $x<_\sigma w'\leq_\sigma y$  satisfying 
$y,w'$ are strongly-connected in $G-L(w')$, we have $y,w'$ are strongly-connected in $G-x$.
This proves, $w$ is indeed the immediate successor of $x$ satisfying 
$y,w$ are strongly-connected in $G-L(w)$.

Let $a$ be the LCA of $x$ and $y$ in $T_{pred}$. So, $a\leq_\sigma x<_\sigma y$.
Let $b$ be the child of $a$ on $a -  y$ path in $T_{pred}$. The vertices $a,b$ are computable in  
$O(1)$ time using LCA/LA data-structures.	
It is easy to verify that $y,b$ are strongly-connected in $G-L(b)$.

We next prove $x <_\sigma b$. Observe descendants of $b$ in $T_{pred}$ are precisely 
the set of vertices in $R(b)$ that are strongly connected to $b$ in $G-L(b)$. 
Thus if $x$ is identical to or a successor of $b$, then $b,x,y$ must be strongly connected 
 in $G-L(b)$ which violates the fact that $x$ is not contained in the sub-tree of $b$.

Finally, we will argue that $w=b$. Let us assume on contrary that $w<_\sigma b$. 
Let $c$ be the immediate predecessor of~$b$, satisfying $x<_\sigma c$ and
$y,c$ are strongly-connected in $G-L(c)$. Observe, $c$ is also the immediate predecessor 
of $b$, satisfying $x<_\sigma c$ and $b,c$ are strongly-connected in $G-L(c)$.
Recall, $a$ is the parent of $b$, so it follows $c\leq_\sigma a$.
This results in a contradiction as $a\leq_\sigma x<_\sigma c$.
Therefore, $w=b$.

The lemma follows from the fact that $b$ is computable in  
$O(1)$ time
\end{proof}

Similar to Lemma~\ref{lemma:find_successor}, we obtain the following lemma.

\begin{lemma}
For any $x,y\in C$ satisfying $y<_{\sigma}x$, we can find the immediate predecessor of $x$
strongly-connected to $y$ in $G-x$ in $O(1)$ time.
\label{lemma:find_predecessor}
\end{lemma}

For any $a\in C$, we store in $H(a)$ the first vertex $b$ in $\sigma$ satisfying 
there is an $a$ to $b$ path in $G-(R(b)\cap L(a))$.
Notice $H(a)\leq_\sigma a$.

Now given a triplet $x,y,z\in C$, we answer $y$ to $z$ reachability query in $G-x$ as follows.

\paragraph{Case 1: $x$ appears between $y$ and $z$ in $\sigma$:}
If $y <_{\sigma} x <_{\sigma} z$, then we answer unreachable as there is no path from $L(x)$ to $R(x)$ in $G-x$.
If $z <_{\sigma} x <_{\sigma} y$, then we proceed as follows.
Use Lemma~\ref{lemma:find_successor} to compute the first successor of $x$, say $y_0$,
such that $y,y_0$ are strongly-connected in $G-x$.
Similarly, compute the first predecessor of $x$, say $z_0$,
such that $z,z_0$ are strongly-connected in $G-x$, using Lemma~\ref{lemma:find_predecessor}.
Now it suffices to check whether or not there is a $y_0$ to $z_0$ path in $G-x$.
If there exists a path $\pi_0$ from $y_0$ to $z_0$ path in $G-x$,
then, by definition of $x_0$ and $y_0$, it follows that $\pi_0$ is disjoint with $R(z_0)\cap L(y_0)$.
Thus, there is path from $y_0$ to $z_0$ (or equivalently, there is a path from $y$ to $z$)
in $G-x$ if and only if $H(y_0)\leq_\sigma z_0$. 

\paragraph{Case 2: $x$ appears before $y$ and $z$ in $\sigma$:}
If $x <_{\sigma} y <_{\sigma} z$, then $\pi[y,z]$ is a path from $y$ to $z$ in $G-x$.
Let us now consider the case $x <_{\sigma} z <_{\sigma} y$.
Use Lemma~\ref{lemma:find_successor} to compute the first successor of $x$, say $y_0$,
such that $y,y_0$ are strongly-connected in $G-x$.
Observe that there is no path from $y_0$ to vertices in $R(x)\cap L(y_0)$ in $G-x$.
Thus, there is path from $y_0$ to $z$ (or equivalently, a path from $y$ to $z$)
in $G-x$ if and only if $y_0\leq_\sigma z$. 

\paragraph{Case 3: $x$ appears after $y$ and $z$ in $\sigma$:}
If $y <_{\sigma} z <_{\sigma}  x$, then $\pi[y,z]$ is a path from $y$ to $z$ in $G-x$.
Let us now consider the case $z <_{\sigma} y <_{\sigma} x$.
Use Lemma~\ref{lemma:find_predecessor} to compute the first predecessor of $x$, say $z_0$,
such that $z,z_0$ are strongly-connected in $G-x$.
Observe that $z_0$ is unreachable from vertices in $R(z_0)\cap L(x)$ in $G-x$.
Thus, there is path from $y$ to $z_0$ (or equivalently, a path from $y$ to $z$)
in $G-x$ if and only if $y\leq_\sigma z_0$.\\

The conditions stated in the above three cases are verifiable in constant time.
Our oracle takes $O(|C|)$ space as it only requires the LCA/LA data-structure 
over trees $T_{pred}$ and $T_{succ}$, and the mapping $H$.
This completes the proof Theorem~\ref{theorem:apsp_cut_set}.

\subsection{Reachability Oracle to Handle Vertex Failures}
\label{section:FTRO_vertex}

We will now present a pairwise reachability oracle for handling vertex failures, and 
later extend this to handle edge failures in the Subsection~\ref{section:FTRO_edge}.
For each pair $p=(s,t)\in \P$ define \emph{cut-vertex-set} for $p$, denoted by $\cv(p)$, as the set of all vertices 
$x$ in $G$ satisfying there is no $s$ to $t$ path in $G-x$.

Let $\alpha\geq 1$ be an integer parameter.
Compute a subset $\Q$ of $\P$ as follows.
Iterate over pairs in $\P$, and include a pair $p$ in $\Q$  if and only if $\cv(p)\setminus \big( \cup_{q\in \Q} \cv(q) \big)$
is larger than $\alpha$.  The following simple observation is due to the fact that each pair in $\Q$ is associated with $\alpha$
distinct cut vertices. 

\begin{observation}
$|\Q|\leq (n/\alpha)$.
\end{observation}

Define $V_{\Q}:=\cup_{q\in \Q}\cv(q)$ to be set of cut vertices for pairs in $\Q$. 
Then for every $p$ not in $\Q$, size of $\cv(p)\setminus V_{\Q}$ is bounded by $\alpha$. 

It turns out that the most non-trivial scenario is when failing vertex lies in $V_{\Q}$.
Let $q_1,\ldots,q_{\ell}$ be pairs in~$\Q$ (here $\ell=|\Q|$), and $V_1,\ldots,V_\ell$ be a partition of $V_\Q$
satisfying $V_i\subseteq \cv(q_i)$. 
For each $p\in \P$ and $i\in [1,\ell]$, let $\first(p,i)$ and $\last(p,i)$ be 
respectively the first and last cut vertices for pair $p$ that lies in set $V_i$.
Since all the cut vertices for any pair, say $(s,t)$, always appear in same 
order on each $s -  t$  path, the notion of first and last cut vertices is well defined.

In order to efficiently handle failures in set $V_i$, for $i\in[1,\ell]$, we present in next lemma a linear size data-structure for 
all-pairs reachability as long as the query pair and the failing vertex are restricted to $V_i$.

\begin{lemma}
For each $i\in[1,\ell]$, we can compute in polynomial time an $O(|V_i|)$ size data-structure that
given any triplet  $x,y,z\in V_i$ reports whether or not there is a $y -  z$ path in $G-x$ in $O(1)$ time.
\label{lemma:DS_Vi}
\end{lemma}

The proof of the above lemma is immediate from Theorem~\ref{theorem:apsp_cut_set}. Now using the
above we show the desired bound on the reachability oracle for handling vertex failures.

Let $F$ be a function that maps each $x\in V_\Q$ to the index $i=F(x)$ satisfying $x\in V_i$.
Our oracle comprises of function $F$;  for each $p\in \P$ and $i\in[1,\ell]$, the vertices $\first(p,i)$ and $\last(p,i)$;
a dictionary ${\D}_{\Q}$ to maintain $V_{\Q}$; and dictionaries ${\D}_{p}$ to maintain $\cv(p)\setminus V_\Q$, for $p\in \P$.

The query algorithm given any failing vertex $x$ and pair $p=(s,t)\in \P$ proceeds as follows.
If $x\notin V_\Q$ then it reports `Reachable' if and only if $x\notin (\cv(p)\setminus V_\Q)$.
Using the dictionaries ${\D}_{\Q}$ and ${\D}_{p}$, this is verifiable in constant time.
If $x\in V_\Q$, then the next step is to extract in constant time index $i=F(x)$ and vertices $a=\first(p,i)$, $b=\last(p,i)$.
By definition, $a$ and $b$ are cut vertices for $p$ as well as $q_i\in \Q$.
Since~$a,b$ are cut vertices for $p$, and $a$ appears before $b$ on any $s -  t$ path, 
$t$ is reachable from $s$ in $G-x$ if and only if $b$ is reachable from $a$ in $G-x$. The later condition is verifiable in 
constant time using Lemma~\ref{lemma:DS_Vi}.

The space complexity of the oracle is $O(n+\ell\cdot |\P|+|\P|\cdot \alpha)$, where $\ell=|\Q|\leq (n/\alpha)$.
Substituting the value of $\alpha$ as $\sqrt{n}$ provides a bound of $O(n+ |\P|\sqrt n)$ on the size of oracle.

This proves the bound of $O(n+ |\P|\sqrt n)$ on the size of pair-wise
reachability oracle for handling a single vertex failure in constant time.

\subsection{Reachability Oracle to Handle Edge Failures}
\label{section:FTRO_edge}

We will now extend the above result to handle edge failures. Let $\C\subseteq E(G)$ be the set of all those edges $e$ that are {\em cut-edge} with respect to some pair $p\in \P$.
Let $\C_0$ be the subset containing those edges $e=(x,y)$ in $\C$ that satisfy that $x$ and $y$ are strongly connected in $G$,
but not in $G\setminus \{e\}$. 

We compute two dictionaries $\D_\C$ and $\D_{\C_0}$ so that given an edge $e\in E(G)$ 
it is verifiable in constant time whether or not $e$ lies in $\C$ and/or $\C_0$.

%If $e=(a,b)$ is a cut-edge for some pair $(s,t)$, then $e=(a,b)$ is also a cut-edge for pair $(a,b)$.
%Indeed, existence of an $a$ to $b$ path in $G-e$ implies existence of an $s$ to $t$ path in $G-e$ as well.
%Thus, we have:
%The following lemma is immediate from definition of cut-edges.

Our main idea for handling edge failures is due to the following simple observation.

\begin{lemma}
For any $e=(a,b)\in \C\setminus \C_0$ and any pair $p\in \P$, 
$e$ is a cut-edge for $p$ if and only if both $a$~and $b$ are cut-vertices for $p$.
\label{lemma:edge-vertex-reduction}
\end{lemma}

\begin{proof}
Let us suppose $e$ is a cut-edge for $p=(s,t)$. Then each $s$ to $t$ path must pass through 
$a$ as well as $b$, which implies $a$ and $b$ are cut-vertices for $p$.

Now let suppose us assume $a,b$ are cut-vertices for $p$.
Since $e=(a,b)$ is cut-edge for some pair $q\in \P$, there is no path from $a$ to $b$ in $G\setminus \{e\}$.
Thus, $a$ and $b$ cannot be strongly-connected in $G$, because otherwise $e$ will lie in set $\C_0$.
This implies $a$ is not reachable from $b$ in $G$, and thus $a$ must precede $b$ on each $s$ to $t$ path.
Further, since there is no path from $a$ to $b$ in $G\setminus \{e\}$, edge $e$ must be a cut-edge for $p$. 
\end{proof}

All that remains is to handle failures in the set $\C_0$.
We construct a new graph $\G$ with $n+|\C_0|$ vertices by splitting each edge $e=(a,b)\in \C_0$ 
by inserting a vertex $v_e$ in between. Therefore, the edge $(a,b)$ in $G$ translates to $(a,v_e,b)$ path in $\G$.
To answer reachability query for a pair $p=(s,t)$ on failure of an edge $w\in \C_0$,
it suffices to check whether or not there is an $s$ to $t$ path in $\G-v_e$.

We provide below a bound on cardinalities of sets $\C$ and $\C_0$.

\begin{lemma}
$|\C_0|\leq 2n$, and $|\C|\leq O(n+\min\{|\P|\sqrt{n}, n \sqrt{|\P|},(n|\P|)^{2/3}\})$.
\label{lemma:size_bound_cut_edges}
\end{lemma}

\begin{proof}
Let $H_0$ be strong-connectivity certificate (or subgraph) of $G$, and $H_{\P}$ be a pair-wise reachability certificate (or subgraph) for pair-set $\P$ (without any failure).
%Note $\C_0\subseteq E(H_0)$, and $\C_\P\subseteq E(H_\P)$.
Then edges in $\C_0$ must be contained in $H_0$, and edges in $\C$ must be contained in $H_\P$.
It is straightforward to compute a strong-connectivity certificate with $2n$ edges, and Abboud and Bodwin~\cite{AB18} showed that
$|E(H_\P)|=O(n+\min\{|\P|\sqrt{n}, n \sqrt{|\P|},(n|\P|)^{2/3}\})$. Thus the claim follows.
\end{proof}

By employing the vertex fault-tolerant reachability oracle result from previous subsection, and using Lemma~\ref{lemma:edge-vertex-reduction}
and Lemma~\ref{lemma:size_bound_cut_edges}, we obtain an $O(n+ |\P|\sqrt{n} + (n|\P|)^{2/3}))$
$=O(n+ |\P|\sqrt{n} ))$ sized data-structure
that for any pair $(s,t)$ in $\P$ and any edge failure $e$, answers reachability query for $(s,t,e)$ in constant time. This completes the proof of Theorem~\ref{theorem:1ftro}.

\section{Dual Fault-tolerant Pairwise Reachability Preserver}
\label{section:dual-ftrs}
In this section, we present a construction of a dual fault-tolerant reachability preserver (2-$\ftrs$) for a given (directed) graph $G=(V,E)$ and a vertex-pair set $\P \subseteq V \times V$. In particular, we prove Theorem~\ref{thm:dual-ftrs}. Here we provide a construction of 2-$\ftrs$ with slack, which together with Lemma~\ref{lem:ftrs-slack} shows Theorem~\ref{thm:dual-ftrs}.

\begin{theorem}
\label{thm:dual-ftrs-slack}
For any directed graph $G=(V,E)$ with $n$ vertices and a set $\P \subseteq V \times V$ of vertex-pairs, there exists a 2-$\ftrs$ with slack having at most $O(n^{4/3} |\P|^{1/3})$ edges. Furthermore, we can find such a subgraph in polynomial time.
\end{theorem}

We will use the notations $H_{(s,t)}, P_{(s,t)}^i, Q_{(s,t),v}^j$ defined in Section~\ref{section:single-pair}. Apart from them, we need one more notation to denote the frequency of a vertex in a set of paths. For any set $B$ of paths in $G$, for any vertex $v \in V$, we use $\freq_B(v)$ to denote the number of paths in $B$ that contains $v$.\\
\textbf{Procedure to Construct a 2-$\ftrs$ with slack for the input $(G,\P)$:}
\begin{enumerate}
\item For each $(s,t)\in \P$, construct $H_{(s,t)}=2\text{-}\ftrs(G,(s,t))$ using the construction provided in Section~\ref{section:single-pair}. (Recall, $H_{(s,t)}$ contains two $s-t$ paths $P_{(s,t)}^1$ and $P_{(s,t)}^2$, that intersect only at $s-t$ cut-edges and $s-t$ cut-vertices.) Define $H_{inter}=\bigcup_{(s,t)\in \P}H_{(s,t)}$.
\item Consider the following family of subsets $\F:=\{V(P^i_{(s,t)}[L]), V(P^i_{(s,t)}[-L]) \mid i\in \{1,2\},(s,t) \in \P\}$, where $L=n^{2/3}|P|^{-1/3}$. Construct a fractional hitting set $S$ of $\F$ (by Theorem~\ref{thm:frac-hitting-set}). Let $\Q:=\{(s,t)\in \P \mid \forall_{i\in \{1,2\}} V(P^i_{(s,t)}[L])\cap S \ne \emptyset \text{ and } V(P^i_{(s,t)}[-L]) \cap S \ne \emptyset\}$.
\item For each $v \in S$, construct a 2-$\ftrs(G,\{v\}\times V)$ and 2-$\ftrs(G,V\times \{v\})$, and take union of all these subgraphs to get a subgraph $H_1$.
\item For each $(s,t) \in \Q$ and $i \in \{1,2\}$, consider the subpaths $P_{(s,t)}^i[L]$ and $P_{(s,t)}^i[-L]$. Let $H_2$ be the graph $\bigcup_{(s,t)\in \P, i\in \{1,2\}}(P_{(s,t)}^i[L] \cup P_{(s,t)}^i[-L])$.
\item Let $B$ be the set of all the paths $Q_{(s,t),v}^i$ (defined in Section~\ref{section:single-pair}) for $(s,t)\in \Q$, $v \in V(\cup_{j \in \{1,2\}}P_{(s,t)}^j[-L])$ and $i \in \{1,2\}$.
\item Initialize an empty set $W$.
\item While $\exists v \in V$ such that $\freq_B(v) \ge \sqrt{L|\Q|}$
\begin{enumerate}
\item Add $v$ to $W$ and remove all the paths that contains $v$ from $B$.
\end{enumerate}
\item For each $w \in W$, construct a 2-$\ftrs(G,\{w\}\times V)$ and 2-$\ftrs(G,V\times \{w\})$, and take union of all these subgraphs to get a subgraph $H_3$.
\item Let $H_4$ be the subgraph obtained by taking the union of all the (remaining) paths in $B$.
\item Return the subgraph $H=H_1\cup H_2 \cup H_3 \cup H_4$.
\end{enumerate}

\paragraph*{Size of the subgraph $H$. }By Theorem~\ref{thm:frac-hitting-set}, $|S|=O(n/L)=O((n|\P|)^{1/3})$. We use Theorem~\ref{theorem:ftrs} to build a 2-$\ftrs(G, \{v\} \times V)$ and 2-$\ftrs(G, V \times \{v\})$. So the size of $H_1$ is $O((n|\P|)^{1/3} \cdot n)$. By the construction (Step 4), $H_2$ is of size at most $O(L \cdot |\P|)$.  

Observe, while defining $B$ (Step 5) before the start of the while loop, we add at most $O(L)$ paths for each $(s,t) \in \Q$. So, $|B|= O(L |\Q|)$. At each iteration of the while loop (Step 7(a)), we remove at least $\sqrt{L|\Q|}$ paths from $B$ and add a new vertex in $W$. At the end of the while loop, for each vertex $v \in V \setminus W$, $\freq_B(v) < \sqrt{L |\Q|}$. So, at the end of the while loop $|W| \le O\Big(\frac{L |\Q|}{\sqrt{L|Q|}}\Big)=O(\sqrt{L|\Q|})$. Again, we use Theorem~\ref{theorem:ftrs} to build a 2-$\ftrs(G, \{w\} \times V)$ and 2-$\ftrs(G, V \times \{w\})$ for each $w \in W$. Thus the size of $H_3$ is at most $O(n \sqrt{L|\Q|})$.

Since $H_4$ is the union of all the leftover paths in $B$ after the while loop, in-degree of each vertex in $H_4$ is at most $\sqrt{L|\Q|}$. So the size of $H_4$ is $O(n \sqrt{L|\Q|})$. Hence, we conclude that the size of the final subgraph $H$ is $O((n|\P|)^{1/3} \cdot n + L \cdot |\Q| + n \sqrt{L|\Q|}) = O(n^{4/3} |P|^{1/3})$, for $L=n^{2/3}|P|^{-1/3}$ (note, $|\Q| \le |\P|\le n^2$).

\paragraph*{Correctness of pairwise 2-{\ftrs} with slack. } 
Next, we would like to claim that the subgraph $H$ is a 2-$\ftrs$ with slack for $(G,\P)$. First, we show that the constructed data structure is a 2-$\ftrs(G,\Q)$. Second, we show that $|\Q| = \Omega(|\P|)$. The second part is quite straightforward. By the definition of fractional hitting set, $S$ intersects all but $1/10$ fraction of the subsets of $\F$. Hence, $|\Q| \ge \frac{3}{5} |\P|$. So it only remains to show the first part.
\begin{lemma}
\label{lem:ftrs-correctness}
The subgraph $H$ is a 2-$\ftrs(G,\Q)$.
\end{lemma}
\begin{proof}
Consider any arbitrary pair $(s,t) \in \Q$. It is immediate from the definition of a 2-$\ftrs$ that any 2-$\ftrs(H_{(s,t)},(s,t))$ is also a 2-$\ftrs(G,(s,t))$. So to prove Lemma~\ref{lem:ftrs-correctness}, it suffices to show that for all $(s,t) \in \Q$, $H$ is a 2-$\ftrs(H_{(s,t)},(s,t))$. From now on we will focus on proving this latter claim, the proof of which is quite similar to that of Lemma~\ref{lem:dual-oracle-correctness}.

Suppose there is an $s-t$ path in $H_{(s,t)} \setminus \{f_1,f_2\}$. Observe, by Claim~\ref{clm:special-path}, there must be a nice $s-t$ path $R$ in $H_{(s,t)} \setminus \{f_1,f_2\}$. By the definition of nice path, $R$ would be of the form $P_{(s,t)}^i[s-u] \circ Q_{(s,t),u'}^{i} \circ P_{(s,t)}^j[u'-t]$, for some $i,j \in \{1,2\}$, where the subpath $Q_{(s,t),u'}^{i}$ starts from the vertex $u$ and is edge-disjoint with $P_{(s,t)}^1$ and $P_{(s,t)}^2$. ($Q_{(s,t),u'}^{i}$ could be empty; in that case, $R=P_{(s,t)}^i$ for some $i \in \{1,2\}$.) By the definition of $\Q$, $P_{(s,t)}^i[L]\cap S \ne \emptyset$ and $P_{(s,t)}^j[-L]\cap S \ne \emptyset$. Now, if $u$ does not lie on $P_{(s,t)}^i[L]$, i.e., the length of the subpath $P_{(s,t)}^i[s-u]$ is at least $L$, then $R$ must pass through a vertex $v \in P_{(s,t)}^i[L] \cap S$. Similarly, if $u'$ does not lie on $P_{(s,t)}^j[-L]$, then $R$ must pass through a vertex $v \in P_{(s,t)}^j[-L] \cap S$. In both the scenarios, for some $v \in S$, there is an $s-v$ and $v-t$ path in $H_{(s,t)} \setminus \{f_1,f_2\}$ (and thus also in $G \setminus \{f_1,f_2\}$). This means, by the construction of $H_1$ (Step 3), there is an $s-v$ and $v-t$ path in $H_1 \setminus \{f_1,f_2\}$ as well.

So let us now focus on the case when $u$ lies on $P_{(s,t)}^i[L]$ and $u'$ lies on $P_{(s,t)}^j[-L]$. By the construction (Step 4), $H_2$ contains all the edges of $P_{(s,t)}^i[s-u]$ and $P_{(s,t)}^j[u'-t]$. By the construction of the set $B$ (Step 5), $Q_{(s,t),u'}^{i}$ belongs to the set $B$. Then either during the execution of the while loop (Step 7), $Q_{(s,t),u'}^{i}$ has been removed due to inclusion of some vertex $w \in V(Q_{(s,t),u'}^{i})$ in the set $W$; or $Q_{(s,t),u'}^{i}$ is part of $B$ after the termination of the while loop. For the first case, there is an $s-w$ and $w-t$ path in $H_{(s,t)} \setminus \{f_1,f_2\}$ (and so in $G\setminus \{f_1,f_2\}$), and thus also in $H_3 \setminus \{f_1,f_2\}$ by the construction of $H_3$ (Step 8). In the second case, $Q_{(s,t),u'}^{i}$ is included in $H_4$ (Step 9), and hence $R=P_{(s,t)}^i[s-u] \circ Q_{(s,t),u'}^{i} \circ P_{(s,t)}^j[u'-t]$ also exists in $H_2 \cup H_4$ (and thus in $H\setminus \{f_1,f_2\}$). This concludes the proof.
\end{proof}

Let us now combine Theorem~\ref{thm:dual-ftrs-slack} with Lemma~\ref{lem:ftrs-slack} to finish the proof of Theorem~\ref{thm:dual-ftrs}.

\begin{proof}[Proof of Theorem~\ref{thm:dual-ftrs}]
For single pair, by~\cite{BCR16}, we have a 2-$\ftrs$ of size $O(n)$ that can be computed in polynomial time. By Theorem~\ref{thm:dual-ftrs-slack}, we have a 2-$\ftrs$ with slack of size $O(n^{4/3} |\P|^{1/3})$. Theorem now directly follows from Lemma~\ref{lem:ftrs-slack}.
\end{proof}

\section{A Generic Construction of Pairwise $k$-{\ftrs}}
\label{section:k-ftrs}
In this section, we show that for any pair-set $\P$, there exists a $k$-$\ftrs(G,\P)$ having 
$\widetilde O(k~ 2^k ~n^{\frac{2k}{k+1}}~ |\P|^\frac{1}{k+1})$ edges. We complement this by 
providing a lower bound of $\Omega(2^{\lfloor k/2-1\rfloor} n\sqrt{|\P|})$, for $k\geq 2$, on the 
extremal size of $k$-$\ftrs(G,\P)$ in Section~\ref{section:lower-bounds}.

Due to work of \cite{BCR16} on single-source $\kftrs$, it follows that there exists a $\kftrs(G,\P)$ with
$O(2^kn|\P|)$ edges, which is good for a small $\P$.
The bound achieved in this section is better than this trivial $O(2^kn|\P|)$ bound
whenever $|\P|=\omega(k n^\frac{k-1}{k} \log n)$, which proves Theorem~\ref{thm:ub-k-ftrs}.

Let us first introduce the following notation. For any graph $G$ and vertices $s,t$, we use $\dist(s,t,G)$ to denote the length of a shortest $s-t$ path in $G$. Now we explain our construction. For each $p=(s,t)\in \P$ and each $F\subseteq E$ of at most $k-1$ failures for which there is a $s-t$ path in $G\setminus F$ (i.e., $\dist(s,t,G \setminus F)$ is finite), 
let $Q^1_{p,F},Q^2_{p,F}$ be two $(s,t)$ paths in $G\setminus F$ intersecting only at the $(s,t)$-cut-edges in $G\setminus F$.
Further, for $p\in \P$ and $i=1,2$, let 
$$\A_i(p) = \{ F\subseteq E ~|~ F\text{ is of size $k-1$, and } |Q^i_{p,F}|\leq \ell\}$$
be the set of those subsets $F$ of $E$ of size at most $k-1$ that satisfy that length of $Q^i_{p,F}$ is bounded by $\ell$.

\paragraph*{Procedure to construct {\kftrs}. }
The algorithm for computing the $\kftrs(G,\P)$ is as follows.
We sample a set $W$ of $\Theta\big (k (n/\ell)\log n\big)$ vertices uniformly at random.
For each $w\in W$, we employ~\cite{BCR16} to compute a graph $H_w$ that is a $\kftrs$ for $G$ with 
respect to $w$ as source as well as sink.
That is, $H_w$ preserves reachability between all pairs in $\{w\}\times V$ and $V\times \{w\}$ as long as the number 
of failures is bounded by $k$.
Due to \cite{BCR16}, each $H_w$ is computable in polynomial time and contains at most $O(2^k n)$ edges.
Set $H$ to be the subgraph of $G$ obtained by taking union of all edges in $H_w$, for $w\in W$.
Finally, for each $p\in \P$ and $F\in \A_i(p)$, we add edges of $Q^i_{p,F}$ to $H$, for $i=1,2$.

\begin{algorithm}[!ht]
\DontPrintSemicolon
\setstretch{1.25}
\BlankLine
Sample a random set $W$ of $\Theta\big (k (n/\ell)\log n\big)$ vertices.\\
\lForEach{$w\in W$}{Compute $\kftrs$ $H_w$ with respect to pairs in $\{w\}\times V$ and $V\times \{w\}$.}
Initialize $H$ to be a graph obtained from $G$  by taking union of all edges in $H_w$, for $w\in W$.\\
\ForEach{$p\in \P$, $F\in \A_i(p)$, $i\in[1,2]$}
%{\lIf{$|Q^i_{p,F}|\leq \ell$}
{Add edges of $Q^i_{p,F}$ to $H$.}
%}
Return $H$.
\caption{Procedure to construct a $\kftrs(G,\P)$}
\label{algo:kftrs}
\end{algorithm}

We want to claim that the subgraph $H$ is a $\kftrs(G,\P)$. 
Let us first provide a bound on size of the set $\A(p):=\A_1(p) \cup \A_2(p)$.

\begin{lemma}
For each $p\in \P$, $|\A(p)|$ is at most $\ell^{k-1}$.
Moreover, $\A(p)$ is computable in $O(\ell^{k-1}m)$ time, where $|E|=m$.
\label{lemma:computing_A(p)}
\end{lemma}

\begin{proof}
Consider a pair $p=(s,t)\in P$. 
Let $B_r$, for $r\geq 0$, be family of all subsets of $E$ of size $r$ satisfying 
$\dist(s,t,G \setminus F)\leq \ell$. 
Observe $B_0$ contains empty-set if and only if $(s,t)$ distance in $G$ is bounded by $\ell$ (i.e., $\dist(s,t,G) \le \ell$).
The~set~$B_{r}$ can be computed from $B_{r-1}$ as follows: For each $F\in B_{r-1}$, an (arbitrary) shortest $s-t$ path $\pi_{G\setminus F}(s,t)$ in $G\setminus F$ and $e\in \pi_{G\setminus F}(s,t)$, 
include $F\cup\{e\}$ in $B_{r}$ if and only if $(s,t)$ distance in $G\setminus (F\cup\{e\})$ is bounded by $\ell$.
This implies that $|B_{k-1}|\leq \ell^{k-1}$, and $B_{k-1}$ can be computed in 
is $O(m \ell^{k-1})$.

For each $F\in \A(p)$, we have $dist(s,t,G-F)$ is bounded by $\ell$, so $\A(p)\subseteq B_{k-1}$, and $|\A(p)|$ is bounded by 
$\ell^{k-1}$. Further $\A(p)$ can be computed in $O(\ell^{k-1}m)$ time by iterating over all pairs $F\in B_{k-1}$ 
and including them in $\A(p)$ if length of $Q^1_{p,F}$ or $Q^2_{p,F}$ is bounded by $\ell$.
\end{proof}

We will prove that $\ell^{k-1}$ is bounded by $O(n^2)$, therefore, it turns out that the set $\A(p)$
is computable in polynomial time.

\begin{lemma}
With high probability following holds:
For each $p=(s,t)\in \P$ and $F\subseteq E$ of at most $k-1$ edges satisfying $\dist(s,t, G\setminus F)$ is finite, 
$W$ has non-empty intersection with $Q^i_{p,F}$ if its length is larger than $\ell$, for $i=1,2$.
\label{lemma:hitting_property}
\end{lemma}

\begin{proof}
Consider a pair $p=(s,t)\in \P$ and a set $F\subseteq E$ of at most $k-1$ edges satisfying $\dist(s,t,G\setminus F)$ is finite.
For~$i=1,2$, if $|Q^i_{p,F}|\geq \ell$,  then the probability $Q^i_{p,F}$ is disjoint with $W$ is at most $n^{-\Theta(k)}$.
The number of subsets of $E$ of size at most $k-1$ is $n^{O(k)}$.
Taking union over each $p\in \P$ and each subset of $E$ of size at most $k-1$, we obtain that with probability at least $(1-1/n^2)$,
for each $p\in \P$ and $F\subseteq E$ of at most $k-1$ edges,
$W$ has non-empty intersection with $Q^i_{p,F}$ if its length is larger than $\ell$, for $i=1,2$.
%As $H_0$ contains $\kftrs$ for pairs in $\{w_{p,F}\}\times V$ and $V\times \{w_{p,F}\}$, for each $(p,F)$ in $\A$,
%with high probability, there exists a path from $s$ to $t$ in $H_0-F$, for each $(p=(s,t),F)$ in $\A$.
\end{proof}

\begin{lemma}
With high probability following holds:
For each $p=(s,t)\in \P$ and $F\subseteq E$ of size at most $k$, there exists an $s-t$ path in $G\setminus F$ if and only if
there exists an $s-t$ path in $H\setminus F$.
\label{lemma:correctness}
\end{lemma}

\begin{proof}
Consider a pair $p=(s,t)\in \P$. Let $F$ be a subset of at most $k-1$ edges, and $e$ be an edge in $E\setminus F$
satisfying $\dist(s,t,G\setminus(F\cup e))$ is finite. 
Observe that at least one of the paths $Q^1_{p,F}$ or $Q^2_{p,F}$ is disjoint with $e$.
%Without loss of generality assume 
Let $i\in\{1,2\}$ be the index satisfying $e\notin Q^i_{p,F}$.
If $|Q^i_{p,F}|\leq \ell$, then $Q^i_{p,F}$ is contained in $H$ which proves the claim.
So let us assume $|Q^i_{p,F}|>\ell$. By Lemma~\ref{lemma:hitting_property}, with 
high probability $Q^i_{p,F}$ has non-empty intersection with $W$.
Let us suppose $w\in W$ lies on $Q^i_{p,F}$.
Then there is a path from $s$ to $w$ and $w$ to $t$ in $G\setminus(F\cup e)$.
As $H$ contains $\kftrs$ for pairs in $\{w\}\times V$ and $V\times \{w\}$, 
there must exist an $s-w$ and $w-t$ path in $H\setminus (F\cup e)$, which proves the claim.
\end{proof}

Lemma~\ref{lemma:correctness}  proves that the subgraph $H$ is a $\kftrs$ for $(G,\P)$. The subgraph $H$ contains
$$O(|\P|\ell^{k} + k\cdot 2^k~(n^2/ \ell)\log n)$$ 
edges.
Choosing $\displaystyle \ell=(k2^k n^2 |\P|^{-1} \log n)^{\frac{1}{k+1}}$ bounds the number of edges in $H$ by
$O(k~ 2^k~n^{\frac{2k}{k+1}}~ |\P|^\frac{1}{k+1} ~\log^{\frac{k}{k+1}} n).$
We thus obtain the following theorem.

\begin{theorem}
There exists a Monte-Carlo algorithm that, for any $k \ge 1$ and any directed graph $G=(V,E)$ with $n$ vertices and a set $\P \subseteq V \times V$ of vertex-pairs, computes a $\kftrs(G,\P)$ in polynomial time  with $\widetilde O(k~ 2^k~n^{\frac{2k}{k+1}}~ |\P|^\frac{1}{k+1})$ edges.
\end{theorem}

\section{Lower Bounds for Pairwise {\ftro} and {\ftrs}}
%\section{Lower Bound Results for Dual Fault-Tolerant Structures}
%\section{Lower Bounds under Dual Failures}
\label{section:lower-bounds}

In this section, we show that the size bound of pairwise 2-$\ftro$ shown in Theorem~\ref{thm:dual-oracle} is optimal upto the word size. First, we construct a hard instance for which the whole graph is required to be stored as a pairwise 2-$\ftrs$. Next, we use that graph to argue that any pairwise 2-$\ftro$ must be of size same as that graph. 

\subsection{A lower bound on the size of subgraph}
\label{subsection:lower_bound_subgraph}
We first provide construction of a directed graph $G=(V,E)$ along with a set $S$ of $\Theta(r)$ vertices such that any 2-$\ftrs(G, S\times S)$ requires $\Omega(nr)$ edges.

Let $n,r,N\geq 1$ be integers satisfying $n=2rN$.
We construct a directed graph  $G=(V,E)$ on $n$ vertices along with a source-set $S$ of $\Theta(r)$
vertices as follows:
We take $2r$ vertex-disjoint directed paths, $P_1,\ldots,P_r$, $Q_1,\ldots,Q_r$, each of length $N$. Let the path $P_i=(p_{1,i},p_{2,i},\ldots,p_{N,i})$ and the path $Q_j=(q_{1,j},q_{2,j},\ldots,q_{N,j})$, for $i,j\in[r]$.
Let $S=\{a_1,\ldots,a_r,b_1,\ldots,b_r\}$ be such that
$P_i$ starts from node $a_i$, and $Q_j$ terminates at node $b_j$, for $i,j\in [r]$.
Further, for each $k \in [N]$, let $A_k=\{p_{k,1},\ldots,p_{k,r}\}$ and  $B_k=\{q_{k,1},\ldots,q_{k,r}\}$.
The vertex set $V(G)$ of our graph $G$ is just the union of vertices on the paths $P_i$'s and $Q_j$'s; 
and the edge set $E(G)$ is union of sets $\big\{(E(P_i)\cup E(Q_j)\mid i,j\in [r]\big\}$ and 
$\big\{(A_k\times B_k) \mid k\in [N]\big\}$.

The graph $G$ has $n$ vertices and $\Omega(nr)$ edges,
since out-degree of vertices in $\bigcup_{k=1}^N A_k$ and
in-degree of vertices in $\bigcup_{k=1}^N B_k$ is at-least $r$.
Next consider a triplet $(i,j,k)\in [r]\times [r]\times [N]$.
We will show that the edge $(p_{k,i},q_{k,j})\in A_k\times B_k$ must be present in any 2-$\ftrs(G,S\times S)$. To prove this consider the source-destination pair: $(a_i,b_j)$, and
let $ f_1=(p_{k,i},p_{k+1,i}),~~ f_2=(q_{k-1,j},q_{k,j})$. Observe, on failure of $f_1$ and $f_2$, 
$R=(a_i=p_{1,i},p_{2,i},\ldots,p_{k,i})\circ (p_{k,i},q_{k,j}) \circ (q_{k,j},q_{k+1,j},\ldots,q_{N,j}=b_j)$
is the unique path from $a_i$ to $b_j$ in $G\setminus \{f_1,f_2\}$.
This shows that for any $k\in [N]$, all the edges in $A_k\times B_k$ must be contained in a 2-$\ftrs(G,S\times S)$, thereby implying the lower bound of $\Omega(Nr^2)=\Omega(n|S|)$ edges.

So we have the following theorem.

\begin{theorem}
\label{thm:lb-ftrs-source}
For every $n,r~(r\leq n)$, there exists an $n$-vertex directed graph $G$ and a set $S$ of $r$ vertices such that any $2$-$\ftrs(G,S\times S)$ requires $\Omega(r n)$ edges.
\end{theorem}

As an immediate corollary of the above theorem, we get Theorem~\ref{thm:lb-ftrs-pair}.

\paragraph*{Lower Bound under multiple failures. }
We now show how to extend the above lower bound result to multiple failures. 
Given parameters $\rho,k\geq 1$, first compute the lower bound 2-{\ftrs} graph $G$ considered
in above construction with $|S|=2r$, where 
$r=2^k \rho$.
We arbitrarily split set $\{a_1,\ldots,a_r\}$  into $\rho$ sets, say $X_1,\ldots,X_\rho$, each of size $2^k$, and similarly 
split  $\{b_1,\ldots,b_r\}$ into $\rho$ sets, say $Y_1,\ldots,Y_\rho$, again each of size $2^k$.
For $j=1$ to $\rho$, 

\begin{itemize}
\item Compute a binary tree $T_{j}$ on $2^{k+1}-1$ new nodes, and merge the $i^{th}$ leaf of
this tree with $i^{th}$ vertex in $X_j$. We assume edges in tree are directed away from the root node
(say $x_{j}$).

 \item Compute a binary tree $T'_{j}$ on $2^{k+1}-1$ new nodes, and merge the $i^{th}$ leaf of
this tree with $i^{th}$ vertex in $Y_j$. We assume edges in tree are directed towards the from root node
(say $y_{j}$).
\end{itemize}

% Observe that for each $p=(a_i,b_j)\in S\times S$ satisfying $i,j\in [1,r]$, we can compute 
%a set $F$ of size $2k$ edge failures in binary trees defined above such that
%there is a path from 
%from 

Set $W=\{x_1,\ldots,x_\rho,y_1,\ldots,y_\rho\}$ to be our new source set.
We will show that all the edges of the resultant graph (say $H$) are essential in $W\times W$ {\ftrs} when the number of failures is
$2k+2$.
Consider a pair $(a,b)$, such that $a\in A_1$ and $b\in B_N$.
Let $(i,j)$ be such that $a\in X_i$ and $b\in Y_j$.
Let $\pi_a$ be path from $x_i$ to $a$ in tree $T_i$, and
$\pi_b$ be path from $b$ to $y_j$ in tree $T'_j$.
Let $F_0$ be set of $2k$ edge failures that comprises of:
(i) all the edges in $T_i$ with exactly one endpoint (i.e. tail of the edge) on $\pi_a$,
(ii) all the edges in $T'_j$ with exactly one endpoint (i.e. head of the edge) on $\pi_b$.
Then each path from $x_i$ to $y_j$ in the surviving graph $H\setminus F_0$ passes through vertices $a,b$.
Therefore,
preserving reachability between pairs in the set $S\times S$ under $2$ failures in $G$,
translates to 
preserving reachability between pairs in the set $W\times W$ under $2k+2$ failures in $H$.
This proves that a $(2k+2)$-{\ftrs} of $H$ with respect to $W\times W$ must contains $\Omega(\rho ~2^{k}~n)$
edges.

We thus conclude with the following theorem.

%Consider the source destination pair $(x_i,y_j)$.

%Now for different choices of 

\begin{theorem}
\label{thm:lb-kftrs-source}
For every $n$, there exists an $n$-vertex directed graph $G$ and a set $S$ of $\rho$ vertices such that any $\kftrs(G,S\times S)$, for an even $k$, requires $\Omega(\rho ~2^{k/2}~n)$ edges.
\end{theorem}

\subsection{A lower bound on the size of oracle}
\label{subsection:lower_bound_oracle}
We now modify the construction of $G$ (used in Section~\ref{subsection:lower_bound_subgraph})  
to obtain a lower bound on the size of dual-fault-tolerant reachability oracle. We will provide a reduction from a closed variant of the well-studied \emph{Index} problem to the problem of constructing 2-$\ftro$ of a graph. 

We consider the following variant of the indexing problem {\textsc{3-D Index}}: Alice is given $N$ matrices $Z_1,\cdots,Z_N$
in $\{0,1\}^{r \times r}$. Bob is given a triplet $(k,i,j)$, where $k\in [N]$, $i,j\in [r]$ and his task is to output $(i,j)^{th}$ bit of the matrix $Z_k$.

Here, we consider the randomized one-way communication complexity of the above problem. More specifically, suppose only Alice can send bits to Bob (but Bob cannot send anything to Alice). Then the question is how many bits Alice needs to send to solve the problem {\textsc{3-D Index}} with probability at least $2/3$. 

In the standard Index problem, Alice is given a string $x \in \{0,1\}^{\ell}$ and Bob is given an index $i \in [\ell]$, and Bob's task is to output the $i^{th}$ bit of $x$. It is well-known from~\cite{KNR99}, the randomized one-way communication complexity of the Index problem is $\Omega(\ell)$. It is easy to see that the standard Index problem is a special case (with $N=\ell$ and $r=1$) of {\textsc{3-D Index}}.

\begin{theorem}[\cite{KNR99}]
\label{thm:indexing-lb}
The randomized one-way communication complexity of {\textsc{3-D Index}} is $\Omega(Nr^2)$.
\end{theorem}

Next, we provide a reduction from {\textsc{3-D Index}} to the problem of constructing 2-$\ftro$. We modify the graph $G$ constructed in Section~\ref{section:lower-bounds} as follows: For each $i,j\in[r]$ and each $k\in[N]$, we keep the edge $(p_{k,i},q_{k,j})\in A_k\times B_k$ in $G$ if and only if the $(i,j)^{th}$ bit of matrix $Z_k$ is one. Now, we claim that for any triplet $(k,i,j)$, where $k \in [N]$ and $i,j \in [r]$, the $(i,j)^{th}$ bit of the matrix $Z_k$ can be computed by deciding whether there is a $a_i-b_j$ path in $G\setminus \{f_1,f_2\}$ for $ f_1=(p_{k,i},p_{k+1,i}),~~ f_2=(q_{k-1,j},q_{k,j})$. This is because, on failure of $f_1$ and $f_2$, there is a path from $a_i$ to $b_j$ in $G$ if and only if $(p_{k,i},q_{k,j})\in E(G)$.

So, given any instance of {\textsc{3-D Index}}, Alice will construct a graph $G$ as above, and consider a source-set $S=\{a_1,\ldots,a_r,b_1,\ldots,b_r\}$ (as in Section~\ref{section:lower-bounds}). Then she will construct a 2-$\ftro(G,S\times S)$ and send that data structure to Bob. Upon receiving the data structure, Bob (given $(k,i,j)$) can now check whether there is a $a_i-b_j$ path in $G\setminus \{f_1,f_2\}$ for $ f_1=(p_{k,i},p_{k+1,i}),~~ f_2=(q_{k-1,j},q_{k,j})$ and accordingly output the $(i,j)^{th}$ bit of the matrix $Z_k$.

Now, since randomized one-way communication complexity of {\textsc{3-D Index}} is $\Omega(Nr^2)$, any 2-$\ftro(G,S\times S)$ must be of size at least $\Omega(Nr^2)=\Omega(nr)$ bits (recall, $n=\Theta(Nr)$). So, we conclude with the following theorem.

\begin{theorem}
\label{thm:lb-oracle-source}
For any positive integers $n,r~(r\leq n)$, there exists an $n$-vertex directed graph with a source set $S$
of $r$ vertices, such that any 2-$\ftro(G,S\times S)$ must be of size $\Omega\big(n|S|\big)$ (in bits).
\end{theorem}

The above theorem is true even if we allow the oracle to err with probability at most $1/3$. As an immediate corollary of the above theorem, we get Theorem~\ref{thm:lb-oracle-pair}.

\subsection{Any lower bound of {\ftrs} implies that for {\ftro}}
\label{sec:lb-ftro-ftrs}
Above we use the hard instance for 2-$\ftrs$ to argue about the same size lower bound for 2-$\ftro$. However, in general, so far we do not know how to turn a lower bound for pairwise $\ftrs$ to that for pairwise $\ftro$. In this section, we establish such a connection. Let ${\cal S}(n,p,k)$ and ${\cal O}(n,p,k)$ be the respective bounds on the extremal size of $\ftrs$ and $\ftro$ 
for $n$ vertex graphs with $|P|=p$ demand pairs under $k$ edge failures.

\begin{theorem}
\label{thm:lb-ftrs-ftro}
%The bounds ${\cal O}(n,p,1)$ and  ${\cal S}(n,p,0)$ are asymptotically equal.
For any $n,k\geq 1$, the functions $\cal O$ and $\cal S$ satisfy: ${\cal O}(2n,p,k)=\Omega(k^{-1}{\cal S}(n,p,k-1))$.
\end{theorem}

\begin{proof}
For a given $n,p$, let $G=(V,E)$ be an $n$ vertex graph, and $\P$ be a pair-set of size $p$, for which
the optimal $\ftrs(G,\P,k-1)$, say $\G$, contains ${\cal S}(n,p,k-1)$ edges.
For each edge $e\in E(\G)$, there must exist a pair, say $P_e=(s,t)\in \P$, and an edge-set, say $C_e$, 
of size at most $k$ satisfying that $C_e$ is a minimal $(s,t)$-cut in $\G$.
Indeed, if $e$ is in none of the minimal cuts of size $k$, then for each $F\subseteq E(\G)$ of size upto $k-1$,
$e$ would not be an essential edge to preserve reachability between pairs $\P$ in graph $\G$, and so
we could eliminate $e$ from $\G$. 
%We compute a subset $\L$ of $E(\G)$ using the following procedure. 
%Initialize $\L$ to empty. Next 
Compute a set $\L$ by
iterating over the edges of $\G$, and including an edge $e$ in $\L$
if and only if $e$ is not contained in cuts $C_{e'}$, for each $e'$ in partially computed set $\L$.
Clearly $|\L|=\Omega(k^{-1}E(\G))$. 
The set $\L$ satisfies the property that 
%Note that 
for each $e\in \L$, %we have $(C_e\setminus\{e\})\cap \L=\emptyset$.
the set $(C_e\setminus\{e\})$ has an empty intersection with $\L$.

We will present an information-theoretic argument by encoding $|\L|$-length binary-vectors
in the fault-tolerant pair-wise reachability oracle of $2n$ vertex graphs. No data structure can store this in $o(|\L|)$ bits.

For each $J\in \{0,1\}^{|\L|}$, construct a graph $\H_J$ on $2n$ vertices as follows.
Initialize $\H_J$ to $\G$. 
For each $v\in V$, add to $\H_J$ another copy of $v$, say $v_0$. 
Finally, for each $e=(x,y)\in \L$ satisfying that $e_J=0$, add edges $(x,y_0)$ and $(y_0,y)$ to $\H_J$.
Now let $\A$ be an optimal $k$-fault-tolerant reachability oracle for graph family $\{G_J~|~J\in \{0,1\}^{\L}\}$
with respect to pair-set $\P$.
Consider an edge $e\in \L$, and the corresponding pair $P_e$ and cut-set $C_e$.
Our proof works by observing that set $C_e\in \L$ is a cut-set (of size at most $k$) for pair $P_e$ in $\H_J$
if an only if $e_J=1$.
Thus by querying reachability after $k$ failures with respect to pairs in $\P$, we can infer the set $J$ through oracle $\A$.
As the set $J$ can have $2^{|L|}$ configurations, by standard information-theoretic arguments, this 
provides a lower bound of $\Omega(|\L|)=\Omega(k^{-1}{\cal S}(n,p,k-1))$ on the size of oracle $\A$.
\end{proof}

As a direct corollary of the above theorem together with the lower bound of the reachability preserver without failure of~\cite{AB18}, we get Theorem~\ref{thm:1-ftro-lb}. By setting $d=2$ in Theorem~\ref{thm:1-ftro-lb}, for $p=O(n)$, we get a lower bound of $\Omega(n^{2/3}p^{1/2})$ on the size of $1$-FTRO.

\paragraph*{Acknowledgments. }The authors would like to thank anonymous reviewers for many helpful comments on a preliminary version of this work, especially for pointing out~\cite{bodwin2020note}.

%\newpage
\bibliographystyle{plainurl}
\bibliography{references}

%\newpage
%\appendix
%\input{deferred-proofs}
\end{document}